\documentclass[11pt]{article}
\usepackage{hyperref}

\usepackage[usenames,dvipsnames]{xcolor}
\definecolor{Gred}{RGB}{219, 50, 54}
\definecolor{Ggreen}{RGB}{60, 186, 84}
\definecolor{Gblue}{RGB}{72, 133, 237}
\definecolor{Gyellow}{RGB}{247, 178, 16}
\definecolor{ToCgreen}{RGB}{0, 128, 0}
\definecolor{myGold}{RGB}{231,141,20}
\definecolor{myBlue}{rgb}{0.19,0.41,.65}
\definecolor{myPurple}{RGB}{175,0,124}

\usepackage{hyperref}
\hypersetup{
			colorlinks=true,
	citecolor=ToCgreen,
	linkcolor=Sepia,
	filecolor=Gred,
	urlcolor=Gred
	}
\usepackage{preamble}

\title{Entanglement is Necessary for Optimal Quantum Property Testing}
\author{
Sebastien Bubeck \\
\texttt{sebubeck@microsoft.com} \\
Microsoft Research
\and Sitan Chen
\thanks{This work was supported in part by a Paul and Daisy Soros Fellowship, NSF CAREER Award CCF-1453261, and NSF Large CCF-1565235. This work was done in part while S.C. was an intern at Microsoft Research AI.}
\\
\texttt{sitanc@mit.edu}\\
MIT 
\and Jerry Li 
\\
\texttt{jerrl@microsoft.com} \\
Microsoft Research
}

\newcommand{\Proj}{\mathbf{\Pi}}
\newcommand{\Sig}{\mathbf{\Sigma}}
\newcommand{\Lam}{\mathbf{\Lambda}}
\newcommand{\mm}{\rho_{\mathsf{mm}}}
\newcommand{\Mul}{\mathop{\text{Mul}}}

\newcommand{\hatM}{\widehat{M}}
\newcommand{\T}{\vec{T}}
\newcommand{\Sym}{\mathcal{S}}
\newcommand{\Wg}{\mathop{\mathrm{Wg}}}

\newcommand{\diag}{\mathop{\text{diag}}}

\newcommand{\calU}{\mathcal{U}}

\newcommand{\A}{\mathbf{A}}
\newcommand{\B}{\mathbf{B}}
\newcommand{\U}{\mathbf{U}}
\newcommand{\V}{\mathbf{V}}
\newcommand{\W}{\mathbf{W}}
\newcommand{\X}{\mathbf{X}}

\newcommand{\N}{\mathbb{N}}

\newcommand{\vh}{\mathbf{h}}

\newcommand{\nullstate}{\boldsymbol{\rho}^{\le N}_0}
\newcommand{\altstate}{\boldsymbol{\rho}^{\le N}_1}
\newcommand{\nullhypo}[1]{p^{\le #1}_0}
\newcommand{\althypo}[1]{p^{\le #1}_1}
\newcommand{\nullmarg}[1]{p^{#1}_0}
\newcommand{\altmarg}[1]{p^{#1}_1}
\newcommand{\pannull}[1]{p^{\le #1}_0}
\newcommand{\panalt}[1]{p^{\le #1}_1}

\begin{document}

\maketitle

\begin{abstract}
	There has been a surge of progress in recent years in developing algorithms for testing and learning quantum states that achieve optimal copy complexity~\cite{o2015quantum,o2016efficient,haah2017sample,o2017efficient,acharya2019measuring,buadescu2019quantum}. Unfortunately, they require the use of entangled measurements across many copies of the underlying state and thus remain outside the realm of what is currently experimentally feasible. A natural question is whether one can match the copy complexity of such algorithms using only independent---but possibly adaptively chosen---measurements on individual copies.

	We answer this in the negative for arguably the most basic quantum testing problem: deciding whether a given $d$-dimensional quantum state is equal to or $\epsilon$-far in trace distance from the maximally mixed state. While it is known how to achieve optimal $O(d/\epsilon^2)$ copy complexity using entangled measurements, we show that with independent measurements, $\Omega(d^{4/3}/\epsilon^2)$ is necessary, even if the measurements are chosen adaptively. This resolves a question posed in~\cite{wright2016learn}. To obtain this lower bound, we develop several new techniques, including a chain-rule style proof of Paninski's lower bound for classical uniformity testing, which may be of independent interest.
\end{abstract}


\section{Introduction}
\label{sec:intro}
This paper considers the problem of \emph{quantum state certification}.
Here, we are given $N$ copies of an unknown mixed state $\rho \in \C^{d \times d}$ and a description of a known mixed state $\sigma$, and our goal is to make measurements on these copies\footnote{Formally, a measurement is specified by a \emph{positive operator-valued measure (POVM)}, which is given by a set of positive-definite Hermitian matrices $\{M_x\}$ summing to the identity, and the probability of observing measurement outcome $x$ is equal to $\Tr(\rho M_x)$. See Definition~\ref{defn:povm} for details.} and use the outcomes of these measurements to distinguish whether $\rho = \sigma$, or if it is $\epsilon$-far from $\sigma$ in trace norm.
An important special case of this is when $\sigma$ is the maximally mixed state, in which case the problem is known as \emph{quantum mixedness testing}.

This problem is motivated by the need to verify the output of quantum computations.
In many applications, a quantum algorithm is designed to prepare some known $d$-dimensional mixed state $\sigma$.
However, due to the possibility of noise or device defects, it is unclear whether or not the output state is truly equal to $\sigma$.
Quantum state certification allows us to verify the correctness of the quantum algorithm.
In addition to this more practical motivation, quantum state certification can be seen as the natural non-commutative analogue of \emph{identity testing} of (classical) probability distributions, a well-studied problem in statistics and theoretical computer science.

Recently,~\cite{o2015quantum} demonstrated that  $\Theta (d / \epsilon^2)$ copies are necessary and sufficient to solve quantum mixedness testing with good confidence.
Subsequently,~\cite{buadescu2019quantum} demonstrated that the same copy complexity suffices for quantum state certification.
Note that these copy complexities are sublinear in the number of parameters in $\rho$, and in particular, are less than the $\Theta (d^2 / \epsilon^2)$ copies necessary to learn $\rho$ to $\epsilon$ error in trace norm~\cite{o2016efficient,haah2017sample}.

To achieve these copy complexities, the algorithms in~\cite{o2015quantum,buadescu2019quantum} heavily rely on entangled measurements.
These powerful measurements allow them to leverage the representation theoretic structure of the underlying problem to dramatically decrease the copy complexity.
However, this power comes with some tradeoffs.
Entangled measurements require that all $N$ copies of $\rho$ are measured simultaneously.
Thus, all $N$ copies of $\rho$ must be kept in quantum memory without any of them de-cohering.
Additionally, the positive-operator valued measure (POVM) elements that formally define the quantum measurement must all be of size $Nd \times Nd$; in particular, the size of the POVM elements scales with $N$.
Both of these issues are problematic for using any of these algorithms in practice~\cite{cotler2020quantum}.
Entangled measurements are also necessary for the only known sample-optimal algorithms for quantum tomography~\cite{o2016efficient,haah2017sample,o2017efficient}.

This leads to the question: can these sample complexities be achieved using weaker forms of measurement?
There are two natural classes of such restricted measurements to consider:
\begin{itemize}
	\item an (unentangled) \emph{nonadaptive measurement} fixes $N$ POVMs ahead of time, measures each copy of $\rho$ using one of these POVMs, then uses the results to make its decision.
	\item an (unentangled) \emph{adaptive measurement} measures each copy of $\rho$ sequentially, and can potentially choose its next POVM based on the results of the outcomes of the previous experiments.
\end{itemize}
It is clear that arbitrarily entangled measurements are strictly more general than adaptive measurements, which are in turn strictly more general than nonadaptive ones.
However, both nonadaptive and adaptive measurements have the advantage that the  quantum memory they require is substantially smaller than what is required for a generic entangled measurement.
In particular, only one copy of $\rho$ need be prepared at any given time, as opposed to the $N$ copies that must simultaneously be created, if we use general entangled measurements.

Separating the power of entangled vs. nonentangled measurements for such quantum learning and testing tasks was posed as an open problem in~\cite{wright2016learn}.
In this paper, we demonstrate the first such separations for quantum state certification, and to our knowledge, the first separation between adaptive measurements and entangled measurements without any additional assumptions on the measurements, for any quantum estimation task.

We first show a sharp characterization of the copy complexity of quantum mixedness testing with nonadaptive measurements:
\begin{theorem}
	If only unentangled, nonadaptive measurements are used, $\Theta(d^{3/2}/\epsilon^2)$ copies are necessary and sufficient to distinguish whether $\rho\in\C^{d\times d}$ is the maximally mixed state, or if $\rho$ has trace distance at least $\epsilon$ from the maximally mixed state, with probability at least $2/3$.
	\label{thm:nonadaptive}
\end{theorem}
\noindent
Second, we show that $\omega(d)$ copies are necessary, even with adaptive measurements. 
We view this as our main technical contribution.
Formally:
\begin{theorem}
	If only unentangled, possibly adaptive, measurements are used, $\Omega(d^{4/3}/\epsilon^2)$ copies are necesssary to distinguish whether $\rho\in\C^{d\times d}$ is the maximally mixed state, or has trace distance at least $\epsilon$ from the maximally mixed state, with probability at least $2/3$.
	\label{thm:main}
\end{theorem}
\noindent
As quantum state certification is a strict generalization of mixedness testing, Theorems~\ref{thm:nonadaptive} and~\ref{thm:main} also immediately imply separations for that problem as well.
Note that the constant $2/3$ in the above theorem statements is arbitrary and can be replaced with any constant greater than $1/2$. We also remark that our lower bounds make no assumptions on the number of outcomes of the POVMs used, which can be infinite (see Definition~\ref{defn:povm}).



\subsection{Overview of our techniques}
In this section, we give a high-level description of our techniques.
We start with the lower bounds.

\paragraph{``Lifting'' classical lower bounds to quantum ones}
Our lower bound instance can be thought of as the natural quantum analogue of Paninski's for (classical) uniformity testing:
\begin{theorem}[Theorem 4, \cite{paninski2008coincidence}]
	$\Omega(\sqrt{d}/\epsilon^2)$ samples are necessary  to distinguish whether a distribution $p$ over $\{1, \ldots, d\}$ is $\epsilon$-far from the uniform distribution in total variation distance, with confidence at least $2/3$.
	\label{thm:paninski}
\end{theorem}
\noindent
At a high level, Paninski demonstrates that it is statistically impossible to distinguish between the distribution $p_0^{\leq N}$ of $N$ independent draws from the uniform distribution, and the distribution $p_1^{\leq N}$ of $N$ independent draws from a random perturbation of the uniform distribution, where the marginal probability of each element in $\{1, \ldots, d\}$ has been randomly perturbed by $\pm \epsilon / d$ (see Example~\ref{example:paninski}).

The hard instance we consider can be viewed as the natural quantum analogue of Paninski's construction.
Roughly speaking, rather than simply perturbing the marginal probabilities of every element in $\{1, \ldots, d\}$, which corresponds to randomly perturbing the diagonal elements of the mixed state, we also randomly rotate it (see Construction~\ref{constr:quantum}).
We note that this hard instance is not novel and has been considered before in similar settings~\cite{o2015quantum,wright2016learn,haah2017sample}.
However, our analysis technique is quite different from previous bounds, especially in the adaptive setting.

The technical crux of Paninski's lower bound is to upper bound the total variation distance between $p_0^{\leq N}$ and $p_1^{\leq N}$ in terms of the $\chi^2$-divergence between the two.
This turns out to have a simple, explicit form, and can be calculated exactly.
This works well because, conditioned on the choice of the random perturbation in $p_1^{\leq N}$, both of the distributions $p_0^{\leq N}$ and $p_1^{\leq N}$ have a product structure, as they consist of $N$ independent samples.

This product structure still holds true in the quantum case when we restrict to non-adaptive measurements.
This allows us to do a more involved version of Paninski's calculation in the quantum case and thus obtain the lower bound in Theorem~\ref{thm:nonadaptive}.

However, this product structure breaks down completely in the adaptive setting, as now the POVMs, and hence, the measurement outcomes that we observe, for the $t$-th copy of $\rho$, can depend heavily on the previous outcomes.
As a result, the $\chi^2$-divergence between the analogous quantities to $p_0^{\leq N}$ and $p_1^{\leq N}$ no longer have a nice, closed form, and it is not clear how to proceed using Paninski's style of argument.

Instead, inspired by the literature on bandit lower bounds~\cite{auer2002nonstochastic,bubeck2012regret}, we upper bound the total variation distance between $p_0^{\leq N}$ and $p_1^{\leq N}$ by the KL divergence between these two quantities.
The primary advantage of doing so is that the KL divergence satisfies the chain rule.
This allows us to partially disentangle how much information that the $t$-th copy of $\rho$ gives the algorithm, conditioned on the outcomes of the previous experiments.

At present, this chain-rule formulation of Paninski's lower bound seems to be somewhat lossy.
Even in the classical case, we need additional calculations tailored to Paninski's instance to recover the $\Omega(\sqrt{d}/\epsilon^2)$ bound for uniformity testing (see Appendix~\ref{app:paninski}), without which our approach can only obtain a lower bound of $\Omega (d^{1/3} /\epsilon^2)$ (see Section~\ref{sec:warmup}).
At a high level, this appears to be why we do not obtain a lower bound of $\Omega(d^{3/2} / \epsilon^2)$ for adaptive measurements.
We leave the question of closing this gap as an interesting future direction.

\paragraph{``Projecting'' quantum upper bounds to classical ones}
While the lower bound techniques we employ are motivated by the lower bounds for classical testing, they do not directly use any of those results.
In contrast, to obtain our upper bounds, we demonstrate a direct reduction from non-adaptive mixedness testing to classical uniformity testing.
The reduction is as follows.
First, we choose a random orthogonal measurement basis.
Measuring $\rho$ in this basis induces some distribution over $\{1, \ldots, d\}$.
If $\rho$ is maximally mixed, this distribution is the uniform distribution.
Otherwise, if it is far from maximally mixed, then by similar concentration of measure phenomena as used in the proof of the lower bounds, with high probability this distribution will be quite far from the uniform distribution in $L_2$ distance.
Thus, to distinguish these two cases, we can simply run a classical $L_2$ uniformity tester~\cite{chan2014optimal,diakonikolas2014testing,canonne2018testing}.
See Appendix~\ref{subsec:tester} for more details.

\paragraph{Concentration of measure over the unitary group}
In both our lower bounds and upper bounds, it will crucial to carefully control the deviations of various functions of Haar random unitary matrices.
In fact, specializations of quantities we encounter have been extensively studied in the literature on quantum transport in mesoscopic systems, namely the conductance of a chaotic cavity \cite{brouwer1996diagrammatic,beenakker1997random,blanter2000shot,khoruzhenko2009systematic,al2009statistics}, though the tail bounds we need are not captured by these works (see Section~\ref{sec:intuition} for more details).
Instead, we will rely on more general tail bounds~\cite{meckes2013spectral} that follow from log-Sobolev inequalities on the unitary group $U(d)$.
%

\subsection{Related work}
The literature on quantum (and classical) testing and learning is vast and we cannot hope to do it justice here; for conciseness we only discuss some of the more relevant works below.

Quantum state certification fits into the general framework of quantum state property testing problems.
Here the goal is to infer non-trivial properties of the unknown quantum state, using fewer copies than are necessary to fully learn the state.
See~\cite{montanaro2016survey} for a more complete survey on property testing of quantum states.
Broadly speaking, there are two regimes studied here: the asymptotic regime and the non-asymptotic regime.

In the asymptotic regime, the goal is to precisely characterize the exponential convergence of the error as $n \to \infty$ and $d, \epsilon$ are held fixed and relatively small.
In this setting, quantum state certification is commonly referred to as \emph{quantum state discrimination}.
See e.g.~\cite{chefles2000quantum,audenaert2008asymptotic,barnett2009quantum} and references within.
However, this allows for rates which could depend arbitrarily badly on the dimension.

In contrast, we work in the non-asymptotic regime, where the goal is to precisely characterize the rate of convergence as a function of $d$ and $\epsilon$.
The closest work to ours is arguably~\cite{o2015quantum} and~\cite{buadescu2019quantum}.
The former demonstrated that the copy complexity of quantum mixedness testing is $\Theta (d / \epsilon^2)$, and the latter showed that quantum state certification has the same copy complexity.
However, as described previously, the algorithms which achieve these copy complexities heavily rely on entangled measurements.

Another interesting line of work focuses on the case where the measurements are only allowed to be Pauli matrices~\cite{flammia2011direct,flammia2012quantum,da2011practical,aolita2015reliable}.
Unfortunately, even for pure states, these algorithms require $\Omega (d)$ copies of $\rho$.
We note in particular the paper of~\cite{flammia2012quantum}, which gives a $\Omega (d)$ lower bound for the copy complexity of the problem, even when the Pauli measurements are allowed to be adaptively chosen.
However, their techniques do not appear to generalize easily to arbitrary adaptive measurements.

A related task is that of quantum tomography, where the goal is to recover $\rho$, typically to good fidelity or low trace norm error.
The paper~\cite{haah2017sample} showed that $O(d^2 \log (d / \epsilon) / \epsilon^2)$ copies suffice to obtain $\epsilon$ trace error, and that $\Omega (d^2 / \epsilon^2)$ copies are necessary.
Independently,~\cite{o2016efficient} improved their upper bound to $O(d^2 / \epsilon^2)$.
These papers, in addition to~\cite{o2017efficient}, also discuss the case when $\rho$ is low rank, where $o(d^2)$ copy complexity can be achieved.
Notably, all the upper bounds that achieve the tight bound heavily require entanglement.
In~\cite{haah2017sample}, they demonstrate that $\Omega (d^3 / \epsilon^2)$ copies are necessary, if the measurements are nonadaptive.
It is a very interesting question to understand the power of adaptive measurements for this problem as well.

Quantum state certification and quantum mixedness testing are the natural quantum analogues of classical identity testing and uniformity testing, respectively, which both fit into the general setting of (classical) distribution testing.
There is again a vast literature on this topic; see e.g.~\cite{canonne2017survey,goldreich2017introduction} for a more extensive treatment of the topic.
Besides the papers covered previously and in the surveys, we highlight a line of work on testing with \emph{conditional sampling oracles}~\cite{canonne2015testing,chakraborty2016power,canonne2014testing,acharya2014chasm,bhattacharyya2018property,kamath2019anaconda}, a classical model of sampling which also allows for adaptive queries.
It would be interesting to see if the techniques we develop here can also be used to obtain stronger lower bounds in this setting.
Adaptivity also plays a major role in property testing of functions~\cite{belovs2016polynomial,chen2017beyond,khot2016n,baleshzar2017optimal,chen2017boolean,belovs2018adaptive}, although these problems appear to be technically unrelated to the ones we consider here.

\subsection{Miscellaneous Notation}
\label{sec:notation}

We gather here useful notation for the rest of the paper.
Let $[d]$ denote the set $\{1,\ldots,d\}$.
Given a finite set $S$, we will use $x\sim_u S$ to denote $x$ sampled uniformly at random from $S$.
Given two strings $s$ and $t$, let $s\circ t$ denote their concatenation.
Given $t>1$ and a sequence $x_1,...,x_{t-1}$, define $x_{<t}\triangleq (x_1,...,x_{t-1})$. We will also sometimes refer to this as $x_{\le t-1}$. Also, let $x_{<1}\triangleq \emptyset$.

Given distributions $P,Q$, the total variation distance between $P$ and $Q$ is $\tvd(P,Q) \triangleq \frac{1}{2}\norm{P - Q}_1$. If $P$ is absolutely continuous with respect to $Q$, let $\frac{\d P}{\d Q}(\cdot)$ denote the Radon-Nikodym derivative. The KL-divergence between $P$ and $Q$ is $\KL{P}{Q}\triangleq \E[x\sim Q]{\frac{\d P}{\d Q}(x)\log\frac{\d P}{\d Q}(x)}$. The chi-squared divergence between $P$ and $Q$ is $\chisq{P}{Q}\triangleq \E[x\sim Q]{\left(\frac{\d P}{\d Q}(x) - 1\right)^2}$.

Let $\norm{\cdot}_1$, $\norm{\cdot}_2$, and $\norm{\cdot}_{HS}$ denote trace, operator, and Hilbert-Schmidt norms respectively.
Let $\mm\triangleq \frac{1}{d}\Id$ denote the maximally mixed state.
Given a matrix $M$, let $\hatM\triangleq M/\Tr(M)$.
Given $\A\in\C^{d\times d}$ and $\pi\in\Sym_n$ with cycle decomposition $(C_1,...,C_m)$, let $\iprod{A}_{\pi} \triangleq \prod^m_{i=1}\Tr(A^{\abs{C_i}})$.

Finally, throughout this work, we will freely abuse notation and use the same symbols to denote probability distributions, their laws, and their density functions.

\paragraph{Roadmap}
The rest of the paper is organized as follows:
\begin{itemize}
	\setlength\itemsep{0.1em}
	\item Section~\ref{sec:strategies}\---- We describe a generic setup that captures Paninski's and our settings as special cases and provide an overview of the techniques needed to show lower bounds in this setup. 
	\item Section~\ref{sec:quantum_prelims}\---- We formalize the notion of quantum property testing via adaptive measurements, define our lower bound instance, and perform some preliminary calculations.
	\item Section~\ref{sec:nonadaptive}\---- Proof of the lower bound in Theorem~\ref{thm:nonadaptive}.
	\item Section~\ref{sec:warmup}\---- As a warmup to the proof of Theorem~\ref{thm:main}, we prove a weaker version of Paninski's lower bound using our chain rule approach.
	\item Section~\ref{sec:proof}\---- Proof of our main result, Theorem~\ref{thm:main}.
	\item Section~\ref{subsec:tails}\---- Proof of certain tail bounds for Haar-random unitary matrices which are crucial to the proofs of Theorems~\ref{thm:nonadaptive} and \ref{thm:main}.
	\item Appendix~\ref{subsec:tester}\---- Proof of the upper bound in Theorem~\ref{thm:nonadaptive}.
	\item Appendix~\ref{app:paninski}\---- A more ad hoc chain rule proof of Paninski's optimal $\Omega(\sqrt{d}/\epsilon^2)$ lower bound.
	\item Appendix~\ref{app:facts}\---- Various helpful technical facts.
\end{itemize}


\section{Lower Bound Strategies}
\label{sec:strategies}

The lower bounds we show in this work are lower bounds on the number of observations needed to distinguish between a simple null hypothesis and a mixture of alternatives. For instance, in the context of classical uniformity testing, the null hypothesis is that the underlying distribution is the uniform distribution over $\brk{d}$, and the mixture of alternatives considered in \cite{paninski2008coincidence} is that the underlying distribution was drawn from a particular \emph{distribution over distributions $p$} which are $\epsilon$-far in total variation distance from the uniform distribution (see Example~\ref{example:paninski}). In our setting, the null hypothesis is that the underlying state is the maximally mixed state $\mm$, and the mixture of alternatives will be a particular \emph{distribution over quantum states $\rho$} which are $\epsilon$-far in trace distance from $\mm$ (see Construction~\ref{constr:quantum}).

Note that in order to obtain dimension-dependent lower bounds, as in classical uniformity testing, it is essential that the alternative hypothesis be a mixture. If the task were instead to distinguish whether the underlying state was $\mm$ or some \emph{specific} alternative state $\rho$, then if we make independent measurements in the eigenbasis of $\rho$, it takes only $O(1/\epsilon^2)$ such measurements to tell apart the two scenarios.

For this reason we will be interested in the following abstraction which contains as special cases both Paninski's lower bound instance for uniformity testing \cite{paninski2008coincidence} and our lower bound instance for mixedness testing, and which itself is a special case of Le Cam's two-point method~\cite{lecam1973convergence}.
We will do this in a few steps.
First, we give a general formalism for what it means to perform possibly adaptive measurements:
\begin{definition}[Adaptive measurements]
	Given an underlying space $\calS$, a natural number $N \in \N$, and a (possibly infinite) universe $\calU$ of measurement outcomes, a \emph{measurement schedule} $A$ using $N$ measurements is any (potentially random) algorithm which outputs $M_1, \ldots, M_N: \calS \to \calU$, where each $M_i$ is a potentially random function.
	We say that $A$ is \emph{nonadaptive} if the choice of $M_i$ is independent of the choice of $M_j$ for all $j \neq i$, and we say $A$ is \emph{adaptive} if the choice of $M_t$ depends only on the outcomes of $M_1, \ldots, M_{t - 1}$ for all $t \in [N]$.
	\label{defn:measure}
\end{definition}
\noindent To instantiate this for the quantum setting, we let the underlying space $\calS$ be the set of mixed states, and we restrict the measurement functions to be (possibly adaptively chosen) POVMs.
See Definition~\ref{def:povm-schedule} for a formal definition.
\begin{definition}
	A \emph{distribution testing task} is specified by two disjoint sets $\calS_0, \calS_1$ in $\calS$.
	For any $N \in \N$, and any measurement schedule $A$, we say that $A$ solves the problem if there exists a (potentially random) post-processing algorithm $f: \calU^N \to \{0, 1\}$ so that for any $\alpha \in \{0 , 1\}$, if $D \in \calS_\alpha$, then
	\[
	\Pr {f \left( M_1 (D) \circ \cdots \circ M_N(D) \right) = \alpha } \geq 2/3 \; ,\]
	where $M_1, \ldots, M_N$ are generated by $A$.
	\label{def:testing}
\end{definition}
\noindent
For instance, to instantiate the quantum mixedness testing setting, we let $\calS$ be the set of mixed states, we let $\calS_0 = \{\mm\}$ be the set containing only $\mm$, the maximally mixed state, and we let $\calS_1 = \{\rho: \| \rho - \mm \|_1 > \epsilon \}$.
Note that the choice of $2/3$ for the constant is arbitrary and can be replaced (up to constant factors in $N$) with any constant strictly larger than $1/2$.
With this, we can now define our lower bound setup:
\begin{definition}[Lower Bound Setup: Simple Null vs. Mixture of Alternatives]
	In the setting of Definition~\ref{def:testing},
	 a \emph{distinguishing task} is specified by a null object $D_0 \in \calS_0$, a set of alternate objects $\{D_\zeta\} \subseteq \calS_1$ parametrized by $\zeta$, and a distribution $\calD$ over $\zeta$.

	For any measurement schedule $A$ which generates measurement functions $M_1, \ldots, M_N$,
	 let $\nullhypo{N} = \nullhypo{N}(A)$ and $\althypo{N}= \althypo{N} (A)$ be distributions over strings $x_{\le N}\in \calU^N$, which we call \emph{transcripts of length $N$}. 
	The distribution $\nullhypo{N}$ corresponds to the distribution of $M_1 (D_0) \circ \cdots \circ M_N (D_0)$.
	The distribution $\althypo{N}$ corresponds to the distribution of of $M_1 (D_\zeta) \circ \cdots \circ M_N (D_\zeta)$, where $\zeta \sim \calD$.
	 	\label{defn:generic}
\end{definition}
\noindent
The following is a standard result which allows us to relate this back to property testing:
\begin{fact}
Let $\calS_0, \calS_1$ be a property, let $N \in \N$, and let $\calA$ be a class of measurement schedules using $N$ measurements.
Suppose that there exists a distinguishing task so that for every $A \in \calA$, we have that $\tvd(\nullhypo{N}(A), \althypo{N}(A)) \leq 1/3$.
Then the distribution testing task cannot be solved with $N$ samples by any algorithm in $\calA$.
	\label{fact:basic-lowerbound}
\end{fact}
\noindent
For the remainder of the paper, we will usually implicitly fix a measurement schedule $A$, and just write $\nullhypo{N}$ and $\althypo{N}$.
The properties that we assume (e.g. adaptive or nonadaptive) of this algorithm should be clear from context, if it is relevant.

We next define some important quantities which repeatedly arise in our calculations:
\begin{definition}
	In the setting of Definition~\ref{defn:generic}, for any $t\in[N]$, define $\nullmarg{t}(\cdot|x_{<t}), \altmarg{t}(\cdot|x_{<t})$ to be the respective conditional laws of the $t$-th entry, given preceding transcript $x_{<t}$.
For any $\zeta$, let $\althypo{N}|\zeta$ be the distribution over transcripts from $N$ independent observations from $D_\zeta$.

	Assume additionally that $\althypo{N}|\zeta$ are absolutely continuous with respect to $\nullhypo{N}$, for every $\zeta \in \mathrm{supp} (\calD)$. 
	Then, there will exist functions 
	$\brc{g^{\zeta}_{x_{<t}}(\cdot)}_{t\in[N], x_{<t}\in\calU^{t-1},\zeta\in\mathrm{supp}(\calD)}$, such that for any $\zeta,t,x_{\le t}$, the Radon-Nikodym derivative satisfies \begin{equation}
		\frac{\d\althypo{t}|\zeta}{\d\nullhypo{t}}(x_{\le t}) = \prod^t_{i=1}\left(1 + g^{\zeta}_{x_{<i}}(x_i)\right).
		\label{eq:mixture}
	\end{equation}
	We refer to the $g^{\zeta}_{x_{<t}}(\cdot)$ functions as \emph{likelihood ratio factors}.
\end{definition}

We emphasize that neither $\nullhypo{N}$ nor any of the alternatives $\althypo{N}|\zeta$ is necessarily a product measure. Indeed, this is one of the crucial difficulties of proving lower bounds in the adaptive setting. In the non-adaptive setting, the picture of Definition~\ref{defn:generic} simplifies substantially:

\begin{definition}[Non-adaptive Testing Lower Bound Setup]
	In this case, in the notation of Definition~\ref{defn:generic}, the measurement schedule $A$ is nonadaptive, so $\nullhypo{N}$ and all $\althypo{N}|\zeta$ are product measures. Consequently, the functions $g^{\zeta}_{x_{<t}}$ will depend only on $t$ and not on the particular transcript $x_{<t}$, so we will denote the functions by $\brc{g^{\zeta}_t(\cdot)}_{t\in[N], \zeta\in\mathrm{supp}(\calD)}$.
	\label{defn:nonadaptive_generic}
\end{definition}

Paninski's lower bound for classical uniformity testing \cite{paninski2008coincidence} is an instance of the non-adaptive setup of Definition~\ref{defn:nonadaptive_generic}:

\begin{example}
	Let us first recall Paninski's construction. 
	Here the set $\calS$ is the set of distributions over $[d]$.
	Uniformity testing is the property $S_0 = \{U\}, S_1 = \{U': \tvd(U, U') \geq \epsilon\}$, where $U$ is the uniform distribution over $[d]$.
	In the classical ``sampling oracle'' model of distribution testing, the measurements $M_i$ simply take a distribution $D\in\calS$ and output an independent sample from $D$. In particular, $\calU = \brk{d}$.
	
	To form Paninski's lower bound instance, take $\calD$ to be the uniform distribution over $\brc{\pm 1}^{d/2}$.
	Let the null hypothesis be $D_0$, and let the set of alternate hypotheses be given by $\{D_z\}_{z \in \brc{\pm 1}^{d/2}}$, where $D_z$ the distribution over $\brk{d}$ whose $x$-th marginal is $D_z(x)= \frac{1}{d} + (-1)^x\cdot \frac{\epsilon}{d}\cdot z_{\lceil x/2\rceil}$ for any $x\in\brk{d}$.
	Clearly $D_z \in \calS_1$ for all $z$.

	There is no obviously no adaptivity in what the tester does after seeing each new sample. So the family of likelihood ratio factors $\brc{g^z_t(\cdot)}$ for which \eqref{eq:mixture} holds is given by 
	\begin{equation}
		g^z_t(x) = g^z(x) \triangleq \epsilon(-1)^x\cdot z_{\ceil{x/2}}.
		\label{eq:paninskig}
	\end{equation}
	\label{example:paninski}
\end{example}

The definition of $\nullhypo{N}, \althypo{N}$ in our proofs will be straightforward (see Construction~\ref{constr:quantum}), and by Fact~\ref{fact:basic-lowerbound}, the key technical difficulty is to upper bound the total variation distance between $\nullhypo{N}, \althypo{N}$ in terms of $N$. After recording some notation in Section~\ref{sec:notation}, in Section~\ref{sec:nonadaptive_sketch}, we overview our approach for doing so in the non-adaptive setting of Definition~\ref{defn:nonadaptive_generic}, and in Section~\ref{sec:adaptive_sketch}, we describe our techniques for extending these bounds to the generic, adaptive setting of Definition~\ref{defn:generic}. 

\subsection{Non-Adaptive Lower Bounds}
\label{sec:nonadaptive_sketch}

It is a standard trick to upper bound total variation distance between two distributions in terms of the $\chi^2$-divergence, which is often more amenable to calculations. These calculations are especially straightforward in the non-adaptive setting  of Definition~\ref{defn:nonadaptive_generic}.

\begin{lemma}
	Let $\nullhypo{N}, \althypo{N}, \calD, \brc{g^{\zeta}_t(\cdot)}_{t\in\N,\zeta\in\mathrm{supp}(\calD)}$ be defined as in Definition~\ref{defn:nonadaptive_generic}. As $\nullhypo{N}$ is therefore a product measure, for every $t\in[N]$ denote its $t$-th marginal by $\nullmarg{t}$. Then \begin{equation}
		\frac{1}{2\ln 2}\tvd\left(\althypo{N},\nullhypo{N}\right)^2 \le \chisq{\althypo{N}}{\nullhypo{N}} \le \max_t \E[\zeta,\zeta']*{\left(1 + \E[x_t\sim\nullmarg{t}]*{g^{\zeta}_t(x_t)g^{\zeta'}_t(x_t)}\right)^N} - 1.
	\end{equation}
	\label{lem:nonadaptive_key_lemma}
\end{lemma}

\begin{proof}
	The first inequality is just Pinsker's and the fact that chi-squared divergence upper bounds KL divergence. For the latter inequality, it will be convenient to define \begin{equation}
		g^{\zeta}_S(x_S) \triangleq \prod_{t\in S}g^{\zeta}_t(x_t).
		\label{eq:gS}
	\end{equation} Then for any $\zeta,\zeta',S$, the product structure implies \begin{equation}
		\E[x_{\le N}\sim\nullhypo{N}]*{g^{\zeta}_S(x_S)g^{\zeta'}_S(x_S)} = \prod_{t\in S}\E[x_t\sim\nullmarg{t}]{g^{\zeta}_t(x_t)g^{\zeta'}_t(x_t)}
		\label{eq:use_product}
	\end{equation} We then get that \begin{align}
		\chisq{\althypo{N}}{\nullhypo{N}} &= \E[x_{\le N}\sim \nullhypo{N}]*{\left(\E[\zeta]*{\prod^N_{t=1}(1 + g^{\zeta}_t(x_t))} - 1\right)^2} = \E[x_{\le N},\zeta,\zeta']*{\sum_{\emptyset\neq S,S'\subseteq[N]}g^{\zeta}_S(x_S)g^{\zeta'}_{S'}(x_{S'})} \\
		&= \E[x_{\le N},\zeta,\zeta']*{\sum_{S\neq \emptyset}g^{\zeta}_S(x_S)g^{\zeta'}_S(x_S)} = \E[\zeta,\zeta']*{\prod^{N}_{t=1}\left(1 + \E[x_t\sim\nullmarg{t}]*{g^{\zeta}_t(x_t)g^{\zeta'}_t(x_t)}\right)} - 1  \\
		&\le \max_t \E[\zeta,\zeta']*{\left(1 + \E[x_t\sim\nullmarg{t}]*{g^{\zeta}_t(x_t)g^{\zeta'}_t(x_t)}\right)^N} - 1,\label{eq:chisq_nonadaptive}
	\end{align}
	where the fourth step follows by \eqref{eq:use_product}, the last step follows by Holder's, and the third step follows by the fact that for $S\neq S'$ and any $\zeta,\zeta'$, \begin{equation}
		\E[x_{\le N}]{g^{\zeta}_S(x_S)g^{\zeta'}_{S'}(x_{S'})} = \prod_{t\in S\cap S'}\E[x_t]{g^{\zeta}_t(x_t)g^{\zeta'}_t(x_t)}\cdot \prod_{t\in S\backslash S'}\E[x_t]{g^{\zeta}_t(x_t)}\cdot \prod_{t\in S'\backslash S}\E[x_t]{g^{\zeta'}_t(x_t)} = 0,
		\label{eq:use_indep}
	\end{equation}
\end{proof}

The upshot of \eqref{eq:chisq_nonadaptive} is that the fluctuations of the quantities $\E[x_t]{g^{\zeta}_t(x_t)g^{\zeta'}_t(x_t)}$ with respect to the randomness of $\zeta,\zeta'$ dictate how large $N$ must be for $\nullhypo{N}$ and $\althypo{N}$ to be distinguishable. 

\begin{example}
	Recalling \eqref{eq:paninskig}, the quantities $\E[x_t]{g^{\zeta}_t(x_t)g^{\zeta'}_t(x_t)}$ take a particularly nice form in Paninski's setting. There we have \begin{equation}
		\E[x_t]{g^{\zeta}_t(x_t)g^{\zeta'}_t(x_t)} = \epsilon^2\cdot \E[x\sim \brk{d}]*{z_{\ceil{x/2}}\cdot z'_{\ceil{x/2}}} = \frac{\epsilon^2}{d}\sum^d_{x=1}\bone{z_{\ceil{x/2}} = z'_{\ceil{x/2}}} = \frac{2\epsilon^2}{d} \iprod{z,z'}
		\label{eq:innerexp}
	\end{equation}

	Because $\iprod{z,z'}$ is distributed as a shifted, rescaled binomial distribution, $\E[x_t]{g^{\zeta}_t(x_t)g^{\zeta'}_t(x_t)}$ has sub-Gaussian tails and fluctuations of order $O(\epsilon^2/\sqrt{d})$, implying that for $N$ as large as $o(\sqrt{d}/\epsilon^2)$, $\chisq{\panalt{N}}{\pannull{N}} = o(1)$. While this is not exactly how Paninski's lower bound was originally proven, concentration of the binomial random variable $\iprod{z,z'}$ lies at the heart of the lower bound and formalizes the usual intuition for the $\sqrt{d}$ scaling in the lower bound: to tell whether a distribution is far from uniform, it is necessary to draw $\Omega(\sqrt{d})$ samples just to see some element of $[d]$ appear twice.
	\label{example:paninski_calc}
\end{example}

In Section~\ref{sec:nonadaptive}, we will show how to use Lemma~\ref{lem:nonadaptive_key_lemma} to prove Theorem~\ref{thm:nonadaptive}. As it turns out, understanding the fluctuations of the random variable $\E[x_t]{g^{\zeta}_t(x_t)g^{\zeta'}_t(x_t)}$ that arises in that setting will be one of the primary technical challenges of this work, both for our adaptive and non-adaptive lower bounds (see Section~\ref{subsec:tails}).

\subsection{Adaptive Lower Bounds}
\label{sec:adaptive_sketch}

As was discussed previously and is evident from the proof of Lemma~\ref{lem:nonadaptive_key_lemma}, the lack of product structure for $\nullhypo{N}$ and $\althypo{N}|\zeta$ in the adaptive setting of Definition~\ref{defn:generic} makes it infeasible to directly estimate $\chisq{\althypo{N}}{\nullhypo{N}}$. Inspired by the literature on bandit lower bounds \cite{auer2002nonstochastic,bubeck2012regret}, we instead upper bound $\KL{\althypo{N}}{\nullhypo{N}}$, for which we can appeal to the chain rule to tame the extra power afforded by adaptivity. To handle the mixture structure of $\althypo{N}$, we will upper bound each of the resulting \emph{conditional} KL divergence terms by their corresponding conditional $\chi^2$ divergence.

First, we introduce some notation essential to the calculations in this work.

\begin{definition}[Key Quantities]
	In the generic setup of Definition~\ref{defn:generic}, for any $x_{\le t}\in\calU^{t}$, define
	\begin{equation}
		\Delta(x_{\le t}) \triangleq \frac{\d\althypo{t}}{\d\nullhypo{t}}(x_{\le t}), \  \phi^{\zeta,\zeta'}_{x_{\le t}} \triangleq \E[x\sim \nullmarg{t}(\cdot|x_{\le t})]*{g^{\zeta}_{x_{\le t}}(x)g^{\zeta'}_{x_{\le t}}(x)}, \ \Psi^{\zeta,\zeta'}_{x_{\le t}} \triangleq \prod^{t}_{i=1}(1 + g^{\zeta}_{x_{<i}}(x_i))(1 + g^{\zeta'}_{x_{<i}}(x_i))
		\label{eq:main_quantities}
	\end{equation}
	\label{defn:main_quantities}
\end{definition}

The following is a key technical ingredient of this work.

\begin{lemma}
	Let $\nullhypo{N}, \althypo{N}, \calD, \brc{g^{\zeta}_{x_{<t}}(\cdot)}$ be defined as in Definition~\ref{defn:generic}. Then 
	\begin{equation}
		\frac{1}{2\ln 2}\tvd\left(\nullhypo{N},\althypo{N}\right)^2 \le \KL{\althypo{N}}{\nullhypo{N}} \le \sum^N_{t=1}\E[x_{<t}\sim \nullhypo{t-1}]*{\frac{1}{\Delta(x_{<t})}\E[\zeta,\zeta'\sim\calD]*{\phi^{\zeta,\zeta'}_{x_{<t}}\cdot \Psi^{\zeta,\zeta'}_{x_{<t}}}}.
		\label{eq:main_chain_eq}
	\end{equation}
	\label{lem:main_chain}
\end{lemma}

\begin{proof}
	The first inequality is Pinsker's. For the second, by the chain rule for KL divergence and the fact that chi-squared divergence upper bounds KL, $\KL{\althypo{(N)}}{\nullhypo{N}}$ can be written as \begin{equation}
		\sum^N_{t=1}\E[x_{<t}\sim \althypo{t-1}]*{\KL{\altmarg{t}(\cdot|x_{<t})}{\nullmarg{t}(\cdot|x_{<t})}} \le \sum^N_{t=1}\E[x_{<t}\sim \althypo{t-1}]*{\chisq{\altmarg{t}(\cdot|x_{<t})}{\nullmarg{t}(\cdot|x_{<t})}}.\label{eq:mainKL_ineq}
	\end{equation} By definition, the conditional densities $p^t_0(\cdot|x_{<t}), p^t_1(\cdot|x_{<t})$ satisfy \begin{equation}
	p^t_i(x_t|x_{<t}) = \frac{p^{\le t}_i(x_{< t}\circ x_t)}{p^{\le t-1}_i(x_{<t})} \ \ \ \text{for} \ i = 0,1.
	\label{eq:bayes}
\end{equation} Therefore, we have: \begin{align}
		\E[x_{<t}\sim \althypo{t-1}]*{\chisq{\altmarg{t}(\cdot|x_{<t})}{\nullmarg{t}(\cdot|x_{<t})}}
		&= \E[x_{<t}\sim \althypo{t-1}]*{\E[x_t\sim \nullmarg{t}(\cdot|x_{<t})]*{\left(\frac{\Delta(x_{<t}\circ x_t)}{\Delta(x_{<t})} - 1\right)^2}} \\
		&= \E[x_{<t}\sim \althypo{t-1}]*{\frac{1}{\Delta(x_{<t})^2}\E[x_t\sim \nullmarg{t}(\cdot|x_{<t})]*{\left(\Delta(x_{<t}\circ x_t) - \Delta(x_{<t})\right)^2}} \\
		&= \E[x_{<t}\sim \nullhypo{t-1}]*{\frac{1}{\Delta(x_{<t})}\E[x_t\sim \nullmarg{t}(\cdot|x_{<t})]*{\left(\Delta(x_{<t}\circ x_t) - \Delta(x_{<t})\right)^2}}\label{eq:plugin}
	\end{align}
	where the first step follows by \eqref{eq:bayes} and the third step follows by a change of measure in the outer expectation.

	By the assumption~\eqref{eq:mixture} and the definition of $\Delta(\cdot)$,
	\begin{equation}
		\Delta(x_{<t}) = \E[z]*{\prod^{t-1}_{i=1}(1 + g^{\zeta}(x_i))}.
		\label{eq:LR}
	\end{equation} 
	This yields
	\begin{align}
		\E[x_t\sim\nullmarg{t}(\cdot|x_{<t})]*{(\Delta(x_{<t}\circ x_t) - \Delta(x_{<t}))^2} &= \E[x_t\sim\nullmarg{t}(\cdot|x_{<t})]*{\E[\calD]*{\prod^{t-1}_{i=1}(1 + g^{\zeta}_{x_{<i}}(x_i))\cdot g^{\zeta}_{x_{<t}}(x_t)}^2} \\
		&= \E[\zeta,\zeta']*{\E[x_t]*{g^{\zeta}_{x_{<t}}(x_t)g^{\zeta'}_{x_{<t}}(x_t)}\prod^{t-1}_{i=1}(1 + g^{\zeta}_{x_{<i}}(x_i))(1 + g^{\zeta'}_{x_{<i}}(x_i))} \\
		&= \E[\zeta,\zeta']*{\phi^{\zeta,\zeta'}_{x_{<t}}\cdot \Psi^{\zeta,\zeta'}_{x_{<t}}},
	\end{align} from which the lemma follows by \eqref{eq:plugin}.
\end{proof}


\section{Unentangled Measurements and Lower Bound Instance}
\label{sec:quantum_prelims}

In this section we provide some preliminary notions and calculations that are essential to understanding the proofs of Theorem~\ref{thm:nonadaptive} and \ref{thm:main}. We first formalize the notion of quantum property testing with unentangled, possibly adaptive measurements in Section~\ref{sec:unentangledtesting}. Then in Section~\ref{sec:instance}, we give our lower bound construction and instantiate it in the generic setup of Definition~\ref{defn:generic}. Finally, in Section~\ref{sec:intuition}, we give some intuition for some of the key quantities that arise.

\subsection{Testing with Unentangled Measurements}
\label{sec:unentangledtesting}

We first formally define the notion of a POVM with possibly infinite outcome set.

\begin{definition}
	Given space $\Omega$ with Borel $\sigma$-algebra $\mathcal{B}(\Omega)$, let $\mu$ be a regular positive real-valued measure $\mu$ on $\mathcal{B}(\Omega)$, and let $M:\Omega\to\C^{d\times d}$ be a measurable function taking values in the set of psd Hermitian matrices. We will denote the image of $x\in\Omega$ under $M$ by $M_x$.

	We say that the pair $(\mu, M)$ specifies a \emph{POVM} $\calM$ if $\int_{\Omega} M\, \d \mu = \Id_{d\times d}$ and, for any $d\times d$ density matrix $\rho$, the map $B\mapsto \int_B \iprod{M_x, \rho} \, \d\mu$ for $B\in\mathcal{B}(\Omega)$ specifies a probability measure over $\Omega$. We call the distribution given by this measure the \emph{distribution over outcomes from measuring $\rho$ with $\calM$}.\footnote{This definition looks diferent from standard ones because we are implicitly invoking the Radon-Nikodym theorem for POVMs on finite-dimensional Hilbert spaces, see e.g. Theorem 3 from \cite{moran2013positive} or Lemma 11 from \cite{chiribella2010barycentric}.}

	Given a POVM $\calM$, we will refer to the space of measurement outcomes as $\Omega(\calM)$.
	\label{defn:povm}
\end{definition}

With no meaningful loss in understanding, the reader may simply imagine that all POVMs mentioned henceforth have finitely many outcomes so that a POVM is simply the data of some finite set of positive semidefinite Hermitian matrices $\{M_x\}_{x\in\Omega}$ for which $\sum_x M_x = \Id_{d\times d}$, though our arguments extend to the full generality of Definition~\ref{defn:povm}.

\begin{definition}
\label{def:povm-schedule}
	Let $N\in\N$. An \emph{unentangled, possibly adaptive POVM schedule} $\calS$ is a type of measurement schedule specified by a (possibly infinite) collection of POVMs $\brc*{\calM^{x_{<t}}}_{t\in[N],x_{<t}\in\calT_t}$ where $\calT_1\triangleq \brc{\emptyset}$, and for every $t>1$, $\calT_t$ denotes the set of all possible transcripts of measurement outcomes $x_{<t}$ for which $x_i\in\Omega(\calM^{x_{<i}})$ for all $1\le i\le t-1$ (recall that $x_{<i}\triangleq (x_1,...,x_{i-1})$).
	The schedule works in the natural manner: at time $t$ for $t = 1, \ldots, N$, given a transcript $x_{< t}\in\calT_t$, it measures the $t$-th copy of $\rho$ using the POVM $\calM^{x_{<t}}$.

If in addition the resulting schedule is also a nonadaptive measurement schedule, we say it is an unentangled, nonadaptive POVM schedule.
\end{definition}

\subsection{Lower Bound Instance}
\label{sec:instance}

Let $\calD$ be the Haar measure over the unitary group $U(d)$. In place of $\zeta$ from Definition~\ref{defn:generic}, we will denote elements from $\calD$ by $\U$. $\Pr[\U]{\cdot}$ and $\E[\U]{\cdot}$ will be with respect to $\calD$ unless otherwise specified.

\begin{construction}
	Let $\X\in\R^{d\times d}$ denote the diagonal matrix whose first $d/2$ diagonal entries are equal to $\epsilon$, and whose last $d/2$ diagonal entries are equal to $-\epsilon$. Let $\X'\triangleq \frac{1}{\epsilon}\X$. Let $\Lam\triangleq \frac{1}{d}(\Id + \X)$. 

	Our lower bound instance will be the distribution over densities $\U^{\dagger}\Lam \U$ for $\U\sim\calD$. We remark that this instance, the quantum analogue of Paninski's lower bound instance \cite{paninski2008coincidence} for classial uniformity testing, has appeared in various forms throughout the quantum learning and testing literature \cite{o2015quantum,wright2016learn,haah2017sample}.

	Given $N\in\N$, define $\nullstate\triangleq \mm^{\otimes N}$ and $\altstate\triangleq \E[\U\sim \calD]{(\U^{\dagger}\Lam \U)^{\otimes N}}$. Take any POVM schedule $\calS = \brc*{\calM^{x_{<t}}}_{t\in[N],x_{<t}\in\calT_t}$. Given $t\le N$, define $\nullhypo{t}$ and $\althypo{t}$ to be the distribution over the measurement outcomes when the first $t$ steps of these POVM schedules are applied to the first $t$ parts of $\nullstate$ and $\altstate$ respectively. Equivalently, $\althypo{t}$ can be regarded as the distribution over sequences of $t$ measurement outcomes arising from first sampling $\U$ according to the Haar measure $\calD$ and then applying the first $t$ steps of POVM schedule $\calS$ to $t$ copies of $\rho\triangleq \U^{\dagger}\Lam\U$.
	\label{constr:quantum}
\end{construction}

\begin{lemma}
	For any POVM $\calM$, define \begin{equation}
		g^{\U}_{\calM}(x) \triangleq \iprod{\hatM^{x_{<t}}_{x}, \U^{\dagger}\X\U}.
		\label{eq:quantum_gdef}
	\end{equation} $\althypo{N}$ is absolutely continuous with respect to $\nullhypo{N}$, and the family of likelihood ratio factors $\brc{g^{\U}_{x_{<t}}(\cdot)}$ for which \eqref{eq:mixture} holds for $\nullhypo{N}$ and $\althypo{N}$ defined in Construction~\ref{constr:quantum} is given by $g^{\U}_{x_{<t}}(\cdot) \triangleq g^{\U}_{\calM^{x_{<t}}}$.
	\label{lem:instantiate_quantum}
\end{lemma}

\begin{proof}
	By taking a disjoint union over $\Omega(\calM^{x_{<t}})$ for all $t\in\N$ and transcripts $x_{<t}$, we can assume without loss of generality that there is some space $\Omega^*$ for which $\Omega(\calM^{x_{<t}})$ is a subspace of $\Omega^*$ for every $t, x_{<t}$. For the product space $(\Omega^*)^N$, equip the $t$-th factor with the $\sigma$-algebra given by the join of all $\sigma$-algebras associated to $\Omega(\calM^{x_{\le t}})$ for transcripts $x_{\le t}$ of length $t$. 

	Then the measures $\mu$ in Definition~\ref{defn:povm} for all POVMs $\calM^{x_{<t}}$ induce a measure $\mu^*$ over $(\Omega^*)^N$. Moreover, by definition, $\nullhypo{N}$ and $\althypo{N}$ correspond to probability measures over $(\Omega^*)^t$ which are absolutely continuous with respect to $\mu^*$.

	Because $\iprod{M_x, \mm} > 0$ for any nonzero psd Hermitian matrix $M_x$, absolute continuity of $\althypo{N}$ with respect to $\nullhypo{N}$ follows immediately.
	
	By the chain rule for Radon-Nikodym derivatives, we conclude that
	\begin{equation}
		\frac{\d\althypo{t}|\U}{\d\nullhypo{t}}(x_{\le t}) = \frac{\prod^{t}_{i=1}\iprod{M^{x_{<i}}_{x_i}, \U^{\dagger}\Lam\U}}{\prod^{t}_{i=1}\frac{1}{d}\Tr(M^{x_{<i}}_{x_i})} = \prod^{t}_{i=1}\iprod{\hatM^{x_{<i}}_{x_i}, \U^{\dagger}(\Id + \X)\U} = \prod^{t}_{i=1} (1 + g^{\U}_{x_{<i}}(x_i))
	\end{equation}
	as claimed.
\end{proof}

For any $\U,\U'\in U(d)$, the quantities $\Psi^{\U,\U'}_{x_{<t}}$ and $\phi^{\U,\U'}_{x_{<t}}$ are given by \eqref{eq:main_quantities}. Given a POVM $\calM$, also define $\phi^{\U,\U'}_{\calM}$ in the obvious way. Lastly, we record the following basic facts:

\begin{fact} 
For any POVM $\calM$,
	\begin{enumerate}[label=(\Roman*)]
		\item $\E[x\sim p]{g^{\U}_{\calM}(x)} = 0$ for any $\U\in U(d)$.
		\label{fact:zero}
		\item For any measurement outcome $x$ and $\U,\U'\in U(d)$, $\abs{g^{\U}_{\calM}(x)} \le \epsilon$ and thus $\phi^{\U,\U'}_{\calM} \le \epsilon^2$. \label{fact:absbound}
	\end{enumerate}
	\label{fact:basic}
\end{fact}

\subsection{Intuition for \texorpdfstring{$\phi^{\U,\U'}_{\calM}$}{phi}}
\label{sec:intuition}

Recall from Example~\ref{example:paninski_calc} that for classical uniformity testing, $\phi^{z,z'} = \frac{2\epsilon^2}{d}\iprod{z,z'}$, and by Lemma~\ref{lem:nonadaptive_key_lemma}, the $O(\epsilon^2/\sqrt{d})$ fluctuations of $\phi^{z,z'}$ as a random variable in $z,z'$ precisely dictate the sample complexity of uniformity testing.

One should therefore think of the distribution of the quantity $\phi^{\U,\U'}_{\calM}$ as a ``quantum analogue'' of the binomial distribution whose fluctuations are closely related to the scaling of the copy complexity of mixedness testing.

As we will show in Theorem~\ref{thm:rho_tail}, $\phi^{\U,\U'}_{\calM}$ has $O(\epsilon^2/d^{3/2})$ fluctuations and concentrates well, from which it will follow by integration by parts that $N$ can be taken as large as $o(d^{3/2}/\epsilon^2)$, yielding the lower bound of Theorem~\ref{thm:nonadaptive}.

To get some intuition for where these $O(\epsilon^2/d^{3/2})$ fluctuations come from, suppose $\calM$ were the orthogonal POVM given by the standard basis. Then \begin{equation}
	\phi^{\U,\U'}_{\calM} = \frac{1}{d}\sum^d_{i=1}\iprod*{\diag(\U^{\dagger}\X\U), \diag(\U'^{\dagger}\X\U')} = \frac{1}{d}\sum^d_{i=1}\epsilon^2 \cdot \delta(\U_i)\cdot\delta(\U'_i),
\end{equation} where \begin{equation}
	\delta(v)\triangleq \sum^{d/2}_{i=1} v^2_i - \sum^d_{i=d/2+1}v^2_i.
	\label{eq:deltadef}
\end{equation} For any fixed $i$, $\U_i,\U'_i$ are independent random unit vectors, and the variance of $\delta(\U_i)\cdot\delta(\U'_i)$ is $O(1/d^2)$ (see Fact~\ref{fact:randomunitvector}). If $\U_1,\U'_1...,\U_d,\U'_d$ were all independent, then $\phi^{\U,\U'}_{\calM}$ would thus have variance $\epsilon^4/d^3$, suggesting $O(\epsilon^2/d^{3/2})$ fluctuations as claimed. Of course we do not actually have this independence assumption; in addition, the other key technical challenges we must face to get Theorem~\ref{thm:rho_tail} are 1) to go beyond just a second moment bound and show sufficiently strong concentration of $\phi^{\U,\U'}_{\calM}$, and 2) to show this is the case for \emph{all POVMs}. We do this in Section~\ref{subsec:tails}.


\section{Proof of Non-Adaptive Lower Bound}
\label{sec:nonadaptive}

In this section we prove Theorem~\ref{thm:nonadaptive} by applying Lemma~\ref{lem:nonadaptive_key_lemma}; the technical crux of the proof (and of our proof of Theorem~\ref{thm:main} in the next section) is the following tail bound, whose proof we defer to Section~\ref{subsec:tails}:

\begin{restatable}{theorem}{rhotail}
	Fix any POVM $\calM$. 
	There exists an absolute constant $c''>0$ such that for any $t >  \Omega(\epsilon^2/d^{1.99})$, we have \begin{equation}
		\Pr[\U,\U'\sim\calD]*{\abs*{\phi^{\U,\U'}_{\calM}} > t} \le \exp\left(-c''\brc*{\Min{\frac{d^3t^2}{\epsilon^4}}{\frac{d^2t}{\epsilon^2}}}\right)
	\end{equation}
	\label{thm:rho_tail}
\end{restatable}

\begin{proof}[Proof of Theorem~\ref{thm:nonadaptive}]
	By Fact~\ref{fact:basic-lowerbound}, it suffices to show that no nonadaptive POVM schedule can solve the distinguishing task given by Construction~\ref{constr:quantum}, unless $N = \Omega (d^{3/2 } / \epsilon^2)$.
	For a non-adaptive POVM schedule $\calS$, let $\{\calM^1,...,\calM^N\}$ denote the sequence of POVMs that are used. Recalling \eqref{eq:quantum_gdef}, the likelihood ratio factors $\brc{g^{\U}_t(\cdot)}_{\U\in U(d), t\in[N]}$ for which \eqref{eq:mixture} holds in the nonadaptive setting of Definition~\ref{defn:nonadaptive_generic} are given by $g^{\U}_{\calM^t}(\cdot)$. Similarly, denote $\phi^{\U,\U'}_{x_{<t}}$ by $\phi^{\U,\U'}_t$.

	By Lemma~\ref{lem:nonadaptive_key_lemma}, we have 
	\begin{equation}
		\frac{1}{2\ln 2}\tvd\left(\althypo{N},\nullhypo{N}\right)^2 \le \max_t \E[\zeta,\zeta']*{\left(1 + \phi^{\U,\U'}_t\right)^N} - 1.
	\end{equation}

	To finish the proof, we will show that \begin{equation}
		\sup_{\calM} \E[\U,\U']*{\left(1 + \phi^{\U,\U'}_{\calM}\right)^N} = 1 + o(1)
	\end{equation} for $N = o(d^{3/2}/\epsilon^2)$, from which the proof is complete by \eqref{eq:chisq_nonadaptive}.

	We would like to apply integration by parts (Fact~\ref{fact:stieltjes}) to the random variable $Z\triangleq 1 + \phi^{\U,\U'}_{\calM}$ and the function $f(Z)\triangleq Z^N$. By Part~\ref{fact:absbound} of Fact~\ref{fact:basic}, this random variable is supported in $[1-\epsilon^2,1+\epsilon^2]$. We can take the parameters in Fact~\ref{fact:stieltjes} as follows: set $a \triangleq \epsilon /(N^{1/2}d^{3/4})$, $b\triangleq 1 + \epsilon^2$, and tail bound function $\tau(x) = \exp\left(-c''\brc*{\Min{\frac{d^3(x - 1)^2}{\epsilon^4}}{\frac{d^2(x - 1)}{\epsilon^2}}}\right)$. Note that for $N = o(d^{3/2}/\epsilon^2)$, $(1+\tau(a))f(a) = 1 + o(1)$. So by Fact~\ref{fact:stieltjes} and Theorem~\ref{thm:rho_tail}, \begin{align}
		\MoveEqLeft \E[\U,\U']*{\left(1 + \phi^{\U,\U'}_{\calM}\right)^N} \\
		&\le 1 + o(1) + \int^{1 + \epsilon^2}_{1 + \epsilon^2/d^{3/2}} Nx^{N-1}\cdot \exp\left(-c''\brc*{\Min{\frac{d^3(x - 1)^2}{\epsilon^4}}{\frac{d^2(x - 1)}{\epsilon^2}}}\right)\ \d x 
		\\
		&\le 1 + o(1) + \int^{1 + \epsilon^2}_{\epsilon^2/d^{3/2}} N(1+x)^{N-1}\left(\exp\left(-\frac{c''d^3 x^2}{\epsilon^4}\right) + \exp\left(-{\frac{c''d^2x}{\epsilon^2}}\right)\right)\ \d x \\
		&\le 1 + o(1) + \int^{\infty}_{0} N(1+x)^{N-1}\left(\exp\left(-\frac{c''d^3 x^2}{\epsilon^4}\right) + \exp\left(-{\frac{c''d^2x}{\epsilon^2}}\right)\right)\ \d x \\
		&= 1 + o(1) + (N/2)! (c''d^4/\epsilon^3)^{-N/2} + N! (c''d^2/\epsilon^2)^N = 1 + o(1),
	\end{align} where the final step uses that $N = o(d^{3/2}/\epsilon^2)$.
\end{proof}


\section{A Chain Rule Proof of Paninski's Theorem}
\label{sec:warmup}

As discussed previously, the proof of Theorem~\ref{thm:nonadaptive} completely breaks down when the POVM schedule $\calS$ is adaptive, so we will instead use the chain rule, via Lemma~\ref{lem:main_chain}, to prove Theorem~\ref{thm:main}.

As a warmup, in this section we will show how to use Lemma~\ref{lem:main_chain} to prove a lower bound for \emph{classical} uniformity testing. As it turns out, it is possible to recover Paninski's optimal $\Omega(\sqrt{d}/\epsilon^2)$ lower bound with this approach, the details of which we give in Appendix~\ref{app:paninski}, but in this section we opt to present a proof which achieves a slightly weaker bound. The reason is that in our proof of Theorem~\ref{thm:paninski_weak}, we will make minimal use of the kind of precise cancellations that would yield a tight bound but which, unfortunately, are specific to the product structure of the distribution of random signs $z$. As such, these steps will be general-purpose enough to extend to the quantum setting where the Haar measure over $U(d)$ enjoys no such product structure.

Specifically, we will use the chain rule to show the following:

\begin{theorem}[Weaker Paninski Theorem]
	$\Omega(d^{1/3}/\epsilon^2)$ samples are necessary to test whether a distribution $p$ is $\epsilon$-far from the uniform distribution.
	\label{thm:paninski_weak}
\end{theorem}

In this section, let $\pannull{N},\panalt{N}$ denote the distributions defined in Example~\ref{example:paninski}. Recalling the notation from Example~\ref{example:paninski} and Definition~\ref{defn:main_quantities}, as well as the identities \eqref{eq:paninskig} and \eqref{eq:innerexp}, we immediately get the following from Lemma~\ref{lem:main_chain}:

\begin{restatable}{lemma}{rewritepan}
	\begin{equation}
		\KL{\panalt{N}}{\pannull{N}}\le \sum^N_{t=1}Z_t \ \ \ \text{for} \ \ \ Z_t \triangleq \E[x_{<t}\sim U^{\otimes t-1}]*{\frac{1}{\Delta(x_{<t})}\E[z,z'\sim\{\pm 1\}^{d/2}]*{\phi^{z,z'} \cdot \Psi^{z,z'}_{x_{<t}}}}.\label{eq:paninski_main}
	\end{equation}\label{lem:rewrite_paninski}
\end{restatable}

We will also need the following two estimates (see below for their proofs).

\begin{lemma}
	For any transcript $x_{<t}$, $\Delta(x_{<t}) \ge \left(1 - \epsilon^2\right)^{(t-1)/2}$.
	\label{lem:denom_paninski}
\end{lemma}

\begin{lemma}
	For any $z,z'\in\{\pm 1\}^{d/2}$, $\E[x_{<t}\sim U^{\otimes t-1}]{(\Psi^{z,z'}(x_{<t}))^2} \le (1 + O(\epsilon^2))^{t-1}$.
	\label{lem:naive_num}
\end{lemma}

We now describe how to use these to bound the summands $Z_t$ in \eqref{eq:paninski_main}. As discussed in Example~\ref{example:paninski_calc}, $\phi^{z,z'} = \frac{2\epsilon^2}{d}\iprod{z,z'}$ has $O(\epsilon^2/\sqrt{d})$ fluctuations. If we pretended $\phi^{z,z'}$ was of this magnitude with probability one, then \begin{equation}
	Z_t\approx O(\epsilon^2/\sqrt{d})\cdot \E[x_{<t}\sim U^{\otimes t-1}]*{\frac{1}{\Delta(x_{<t})}\E[z,z'\sim\{\pm 1\}^{d/2}]*{\Psi^{z,z'}_{x_{<t}}}} = O(\epsilon^2/\sqrt{d}),
	\label{eq:approx}
\end{equation} where the last step follows because $\Delta(x_{<t})^2 = \E[z,z']{\Psi^{z,z'}_{x_{<t}}}$ and the likelihood ratio between two distributions always integrates to 1. Then by \eqref{eq:paninski_main} we would in fact even recover Theorem~\ref{thm:paninski}.

Unfortunately, in reality $\phi^{z,z'}$ can be as large as order $\epsilon^2$, albeit with exponentially small probability, so instead we will partition the space of $z,z'\in\{\pm 1\}^{d/2}$ into those for which $\phi^{z,z'}$ is either less than some threshold $\tau$ or greater. When $\phi^{z,z'} \le \tau$, we can bound the total contribution to $Z_t$ of such $z,z'$ by $\tau$. When $\phi^{z,z'} > \tau$, we will use the pointwise estimates from Lemmas~\ref{lem:denom_paninski} and~\ref{lem:naive_num} and argue that because $\Pr{\phi^{z,z'} > \tau}$ is so small, these $z,z'$ contribute negligibly to $Z_t$. The reason we only get an $\Omega(d^{1/3}/\epsilon^2)$ lower bound in the end is that we must take $\tau$ slightly larger than the fluctuations of $\phi^{z,z'}$ to balance the low probability of $\phi^{z,z'}$ exceeding $\tau$ with the pessimistic pointwise estimates of Lemmas~\ref{lem:denom_paninski} and~\ref{lem:naive_num}.

\begin{proof}[Proof of Theorem~\ref{thm:paninski_weak}]
We fill in the details of the strategy outlined above. 
We will use Fact~\ref{fact:basic-lowerbound} with the construction in Example~\ref{example:paninski}.
Given a transcript $x_{<t}$ and $z,z'\in\{\pm 1\}^{d/2}$, let $\bone{\calE^{z,z'}(\tau)}$ denote the indicator of whether $\phi^{z,z'}>\tau$. We have that \begin{align}
	\E[z,z']*{\Psi^{z,z'}_{x_{<t}}\cdot \phi^{z,z'}} &= \E[z,z']*{\Psi^{z,z'}_{x_{<t}}\cdot \phi^{z,z'} \cdot \left(\bone{\calE^{z,z'}(\tau)} + \bone{\calE^{z,z'}(\tau)^c}\right)} \\
	&\le \epsilon^2\cdot \E[z,z']*{\Psi^{z,z'}_{x_{<t}}\cdot \bone{\calE^{z,z'}(\tau)}} + \tau\cdot \E[z,z']*{\Psi_{x_{<t}}^{z,z'}\cdot \bone{\calE^{z,z'}(\tau)^c}} \\
	&\le \epsilon^2\cdot\underbrace{\E[z,z']*{\Psi^{z,z'}_{x_{<t}}\cdot \bone{\calE^{z,z'}(\tau)}}}_{\circled{B}_{x_{<t}}} + \tau\cdot \underbrace{\E[z,z']*{\Psi_{x_{<t}}^{z,z'}}}_{\circled{G}_{x_{<t}}},
\end{align} where in the second step we used Part~\ref{fact:absbound} of Fact~\ref{fact:basic}. Note that for any transcript $x_{<t}$, $\Delta(x_{<t})^2 = \E[z,z']{\Psi^{z,z'}_{x_{<t}}} = \circled{G}_{x_{<t}}$, so by this and the fact that the likelihood ratio between two distributions always integrates to 1, 
\begin{equation}
	\E[x_{<t}\sim U^{\otimes t-1}]*{\frac{1}{\Delta(x_{<t})}\cdot \circled{G}_{x_{<t}}} = \E[x_{<t}]{\Delta(x_{<t})} = 1.\label{eq:one_paninski}
\end{equation}

We conclude that \begin{align}
	Z_t &\le \epsilon^2\cdot \E[x_{<t}\sim U^{\otimes t-1}]*{\frac{1}{\Delta(x_{<t})}\cdot \circled{B}_{x_{<t}}} + \tau\cdot \E[x_{<t}\sim U^{\otimes t-1}]*{\frac{1}{\Delta(x_{<t})}\cdot \circled{G}_{x_{<t}}} \\
	&\le \epsilon^2 \cdot (1 + \epsilon^2)^{(t-1)/2} \E[x_{<t}]{\circled{B}_{x_{<t}}} + \tau,
\end{align} where the second step follows by Lemma~\ref{lem:denom_paninski} and \eqref{eq:one_paninski}. It remains to show that $\tau$ is the dominant quantity above, for appropriately chosen $\tau$.

Pick $\tau = \Omega(\epsilon^2/d^{1/3})$. To upper bound $\E[x_{<t}]{\circled{B}_{x_{<t}}}$, first apply Cauchy-Schwarz to get \begin{align}
	\E[x_{<t}]{\circled{B}_{x_{<t}}} &\le \E[x_{<t},z,z']*{\left(\Psi^{z,z'}_{x_{<t}}\right)^2}^{1/2}\cdot \Pr[x_{<t},z,z']*{\phi^{z,z'} > \tau}^{1/2} \\
	&\le (1 + O(\epsilon^2))^{(t-1)/2} \cdot \exp(-\Omega(d^{1/3})),
\end{align} where the second step follows by Lemma~\ref{lem:naive_num}, \eqref{eq:innerexp}, and standard binomial tail bounds. For $t = o(d^{1/3}/\epsilon^2)$, this quantity is indeed negligible, concluding the proof that $\circled{*} \le O(\epsilon^2/d^{1/3})$ and, by Lemma~\ref{lem:rewrite_paninski}, that $\chisq{\althypo{N}}{\nullhypo{N}} = o(1)$ for $N = o(d^{1/3}/\epsilon^2)$.
\end{proof}

\paragraph{Deferred Proofs}

\begin{proof}[Proof of Lemma~\ref{lem:denom_paninski}]
For any $x_{<t}\in[d]^{t-1}$, we have that \begin{align*}
	\E[z]*{\prod^{t-1}_{i=1}(1 + g^z(x_i))} &\ge \left(\prod_{z\in \{\pm 1\}^{d/2}}\prod^{t-1}_{i=1}(1 + g^z(x_i))\right)^{2^{-d/2}} \\
	&= \left(\prod_{z\in \{\pm 1\}^{d/2}}\prod^{t-1}_{i=1}(1 + g^z(x_i))^{1/2}(1 + g^{-z}(x_i))^{1/2}\right)^{2^{-d/2}} \\
	&= (1 - \epsilon^2)^{(t - 1)/2},
\end{align*} where in the first step we used AM-GM, in the second step we used the fact that if $z$ is chosen uniformly at random from $\{\pm 1\}^{d/2}$, then $-z$ is also distributed according to the uniform distribution over $\{\pm 1\}^{d/2}$, and in the third step we used that for any $x$, $(1 + g^z(x))(1+g^{-z}(x)) = 1 - \epsilon^2$.
\end{proof}

\begin{proof}[Proof of Lemma~\ref{lem:naive_num}]
Note that by both parts of Fact~\ref{fact:basic}, \begin{equation}
	\E[x\sim U]{(1 + g^z(x))^2(1+g^{z'}(x))^2} = 1 + \E[x]{g^z(x)g^{z'}(x)} \le 1 + O(\epsilon^2).
\end{equation} Writing \begin{equation}
	\E[x_{<t}]{\Psi^{z,z'}_{x_{<t}}} \le \E[x_{<t-1}]{\Psi^{z,z'}_{x_{<t-1}}}\cdot (1 + O(\epsilon^2)),
\end{equation} we see that the claim follows by induction on $t$.
\end{proof}

\paragraph{Parallels to Proof of Theorem~\ref{thm:main}}

Lastly, we comment on how these ingredients carry over to our proof of Theorem~\ref{thm:main}. Lemma~\ref{lem:rewrite_paninski} translates verbatim to the quantum setting (see Lemma~\ref{lem:rewrite}), as does the final part of the proof where we partition based on the value of $\phi^{z,z'}$.

Lemma~\ref{lem:denom} will be the quantum analogue of Lemma~\ref{lem:denom_paninski}, and its proof uses a similar trick of AM-GM plus averaging with an involution.

Lemma~\ref{lem:psi_main} will be the quantum analogue of Lemma~\ref{lem:naive_num}. Unfortunately, as we will see later in Section~\ref{sec:proof}, an analogously naive bound will not suffice in our proof of Theorem~\ref{thm:main}. The workaround is somewhat technical, and we defer the details to Lemma~\ref{lem:psi_main} and the discussion preceding it.

Finally, as in Section~\ref{sec:nonadaptive_sketch}, the central technical ingredient in the proof of Theorem~\ref{thm:paninski_weak} is the concentration of $\phi^{z,z'}$. Analogously, in the proof of Theorem~\ref{thm:main}, we will need sufficiently strong tail bounds for $\phi^{\U,\U'}_{\calM}$, which we show in Theorem~\ref{thm:rho_tail}.


\section{An Adaptive Lower Bound for Mixedness Testing}
\label{sec:proof}

In this section we prove our main result, Theorem~\ref{thm:main}.

First, recalling the notation from Construction~\ref{constr:quantum} and Definition~\ref{defn:main_quantities}, as well as the identity \eqref{eq:quantum_gdef}, we immediately get the following from Lemma~\ref{lem:main_chain}:

\begin{lemma}
	\begin{equation}
		\KL{\althypo{N}}{\nullhypo{N}} \le \sum^N_{t=1}Z_t \ \ \ \text{for} \ \ \ Z_t\triangleq \E[x_{<t}\sim \nullhypo{t-1}]*{\frac{1}{\Delta(x_{<t})}\E[\U,\U']*{\Psi^{\U,\U'}_{x_{<t}}\cdot \phi^{\U,\U'}_{x_{<t}}}}
		\label{eq:quantum_main}
	\end{equation}
	\label{lem:rewrite}
\end{lemma}

Take any $t\le N$. To bound $Z_t$ in \eqref{eq:quantum_main}, we first estimate the likelihood ratio $\Delta$ for an arbitrary transcript, in analogy with Lemma~\ref{lem:denom_paninski} from Section~\ref{sec:warmup} respectively:

\begin{lemma}
For any transcript $x_{<t}$, $\Delta(x_{<t}) \ge \left(1 - O(\epsilon^2/d)\right)^{t-1}$.
\label{lem:denom}
\end{lemma}

\begin{proof}
	Recall \eqref{eq:LR}. By convexity of the exponential function and the fact that $1 + g^{\U}_{x_{<i}}(x_i)>0$ for all $\U,i,x_i$, \begin{equation}
		\Delta(x_{<t})\ge \exp\left(\E[U\sim\calD]*{\sum^{t-1}_{i=1}\ln\left(1 + g^{\U}_{x_{<i}}(x_i)\right)}\right) = \prod^{t-1}_{i=1}\exp\left(\E[\U\sim\calD]*{\ln(1 + g^{\U}_{x_{<i}}(x_i))}\right).\label{eq:sub}
	\end{equation} 
	Define the unitary block matrix $\T = \begin{pmatrix}
		\mathbf{0} & \Id_{d/2} \\
		\Id_{d/2} & \mathbf{0}.
	\end{pmatrix}$ As $\calD$ is invariant with respect to left-multiplication by $\T\in U(d)$, for all $i<t$ we have that \begin{align}
		\exp\left(\E[\U\sim\calD]*{\ln(1 + g^{\U}_{x_{<i}}(x_i))}\right) &= \exp\left(\frac{1}{2}\E[\U\sim\calD]*{\ln(1 + g^{\U}_{x_{<i}}(x_i)) + \ln(1 + g^{\T\U}_{x_{<i}}(x_i))}\right) \\
		&= \exp\left(\frac{1}{2}\E[\U\sim\calD]*{\ln(1 + g^{\U}_{x_{<i}}(x_i)) + \ln(1 - g^{\U}_{x_{<i}}(x_i))}\right) \\
		&= \exp\left(\frac{1}{2}\E[\U\sim\calD]*{\ln(1 - g^{\U}_{x_{<i}}(x_i)^2)}\right) \\
		&\ge 1 + \frac{1}{2}\E[\U\sim\calD]*{\ln(1 - g^{\U}_{x_{<i}}(x_i)^2)} \\
		&\ge 1 - \E[\U\sim\calD]*{g^{\U}_{x_{<i}}(x_i)^2)}\label{eq:oneminus}
	\end{align} 
	where the second step follows from the fact that $\T^{\dagger}\X\T = -\X$, the fourth step follows by the elementary inequality $\exp(x)\ge 1 + x$ for all $x$, and the fifth inequality follows by the elementary inequality $\log(1 - x)\ge -2x$ for all $0\le x < 1/2$.

	Finally, note that for any trace-one psd matrix $M$, we may write $M = \sum\lambda_i v_iv_i^{\dagger}$, and for any unit vector $v\in\C^n$, $\E[\U]{\iprod{vv^{\dagger},\U^{\dagger}\X\U}^2} = O(\epsilon^2/d)$. So \begin{equation}
		\E[U]{\iprod{M,\U^{\dagger}\X\U}^2} = \sum_{i,j} \lambda_i\lambda_j \E*{\iprod{v_iv_i^{\dagger},\U^{\dagger}\X\U}\iprod{v_jv_j^{\dagger},\U^{\dagger}\X\U}} \le O(\epsilon^2/d)\cdot \left(\sum_i \lambda_i\right)^2 = O(\epsilon^2/d),
	\end{equation}
	 where the second step follows by Cauchy-Schwarz. From this we conclude that $\E[\U\sim\calD]{g^{\U}_{x_{<i}}(x_i)^2} \le O(\epsilon^2/d)$ for all $i,x_{<i},x_i$, and the lemma follows by \eqref{eq:sub} and \eqref{eq:oneminus}.
\end{proof}

Next, in analogy with Lemma~\ref{lem:naive_num}, we would like to control the expectation of $(\Psi^{\U,\U'}_{x_{<t}})^2$. We remark that like in the proof of Lemma~\ref{lem:naive_num}, one can obtain a naive estimate of $(1 + O(\epsilon^2))^{t-1}$ using just Fact~\ref{fact:basic}, but unlike in the proof of Theorem~\ref{thm:paninski_weak}, such a bound would not suffice here. Instead, we will need the following important moment bound, whose proof we defer to Section~\ref{subsec:tails}:

\begin{restatable}{theorem}{moments}
	For any POVM $\calM$, let $p$ denote the distribution over outcomes from measuring $\mm$ with $\calM$, and let $\gamma > 0$ be an absolute constant. Define the random variable \begin{equation}
		K^{\U,\U'}_{\calM} \triangleq \E[x\sim p]*{\left(g^{\U}_{\calM}(x) + g^{\U'}_{\calM}(x)\right)^2}
	\end{equation} Then for any $n = o(d^2/\epsilon^2)$, we have that \begin{equation}
		\E[{\U,\U'}]*{\left(1 + \gamma\cdot K^{\U,\U'}_{\calM}\right)^n} \le \exp(O(\gamma n \epsilon^2/d))
		\label{eq:Kdef}
	\end{equation}
	\label{thm:moments}
\end{restatable}

We will use this and a series of invocations of Holder's to prove the following sufficiently strong generalization of Lemma~\ref{lem:naive_num}:

\begin{lemma}
	Suppose $t = o(d^2/\epsilon^2)$. Then $\E[x_{<t},\U,\U']{(\Psi^{\U,\U'}_{x_{<t}})^2} \le \exp(O(t\cdot \epsilon^2/d))$.\label{lem:psi_main}
\end{lemma}

\begin{proof}
	Consider any $a,b\in\Z$ for which $a\ge b$ and $a \ge 2$. For any $x_{<t-1}$, let $p$ denote the distribution over measurement outcomes when the POVM $\calM^{x_{<t-1}}$ is applied to $\mm$. We have by Part~\ref{fact:absbound} of Fact~\ref{fact:basic} that
	\begin{equation}
		\E[x\sim p]{g^{\U}_{x_{<t-1}}(x)^a\cdot g^{\U'}_{x_{<t-1}}(x_t)^b} \le \epsilon \E[x\sim p]{g^{\U}_{x_{<t-1}}(x)^2}.
	\end{equation} Recalling Part~\ref{fact:zero} of Fact~\ref{fact:basic}, we conclude that for any $x_{<t-1}$ and constant degree $c\ge 2$,
	\begin{multline}
		\E[x\sim p]*{(1 + g^{\U}_{x_{<t-1}}(x))^c(1 + g^{\U'}_{x_{<t-1}}(x))^c} \\ 
		\le 1 + O_c(\E[x\sim p]{g^{\U}_{x_{<t-1}}(x)^2}) + O_c(\E[x\sim p]{g^{\U'}_{x_{<t-1}}(x)^2}) + O_c(\phi^{\U,\U'}_{x_{<t-1}}) \triangleq 1 + Z_{x_{<t-1}}^{\U,\U'}(c).\label{eq:gc}
	\end{multline}
	By abuse of notation, for POVM $\calM$, define $Z^{\U,\U'}_{\calM}(c)$ in the obvious way.

	For $\alpha_i \triangleq 2\cdot \left(\frac{t-1}{t-2}\right)^i$, we have that \begin{align}
		\MoveEqLeft \E[x_{<t},\U,\U']*{\left(\Psi^{\U,\U'}_{x_{<t}}\right)^{\alpha_i}} \\
		&\le \E[x_{<t-1},\U,\U']*{\left(\Psi^{\U,\U'}_{x_{<t-1}}\right)^{\alpha_i} \cdot \left(1 + Z^{\U,\U'}_{x_{<t-1}}(\alpha_i)\right)} \label{eq:apply_gc} \\
		&\le \E[x_{<t-1},\U,\U']*{\left(\Psi^{\U,\U'}_{x_{<t-1}}\right)^{\alpha_i (t-1)/(t-2)}}^{(t-2)/(t-1)}\cdot \E[x_{<t-1},\U,\U']*{\left(1 + Z^{\U,\U'}_{x_{<t-1}}(\alpha_i)\right)^{t-1}}^{1/(t-1)}\label{eq:holder} \\
		&\le \E[x_{<t-1},\U,\U']*{\left(\Psi^{\U,\U'}_{x_{<t-1}}\right)^{\alpha_{i+1} (t-1)/(t-2)}}\cdot \E[x_{<t-1},\U,\U']*{\left(1 + Z^{\U,\U'}_{x_{<t-1}}(\alpha_i)\right)^{t-1}}^{1/(t-1)}.
	\end{align} where \eqref{eq:apply_gc} follows by \eqref{eq:gc}, and \eqref{eq:holder} follows by Holder's. Unrolling this recurrence, we conclude that
	\begin{align}
		\E[x_{<t},\U,\U']*{\left(\Psi^{\U,\U'}_{x_{<t}}\right)^2} &\le \prod^{t-1}_{i=1}\E[x_{<i},\U,\U']*{\left(1 + Z^{\U,\U'}_{x_{<i}}(\alpha_{t-1-i})\right)^{t-1}}^{1/(t-1)}\label{eq:unroll} \\
		&\le \prod^{t-1}_{i=1}\E[x_{<i},\U,\U']*{\left(1 + Z^{\U,\U'}_{x_{<i}}(2e)\right)^{t-1}}^{1/(t-1)},\label{eq:2e} \\
		&\le \sup_{\calM}\E[\U,\U']*{\left(1 + Z^{\U,\U'}_{\calM}(2e)\right)^{t-1}}
	\end{align} where \eqref{eq:2e} follows by the fact that for $1\le i \le t - 1$, $\alpha_{t-1-i} \le 2\left(1 + \frac{1}{t-2}\right)^{t-2} \le 2e$, and the supremum in the last step is over all POVMs $\calM$. The proof is complete upon noting that $Z^{\U,\U'}_{\calM}$ is at most a constant multiple of $K^{\U,\U'}_{\calM}$ defined in \eqref{eq:Kdef} and invoking Theorem~\ref{thm:moments}.
\end{proof}

We can now complete the proof of Theorem~\ref{thm:main}. Note that the following argument is very similar to the argument we used to complete the proof of Theorem~\ref{thm:paninski_weak}.

\begin{proof}[Proof of Theorem~\ref{thm:main}]
Given a transcript $x_{<t}$ and $\U,\U'\in U(d)$, let $\bone{\calE_{x_{<t}}^{\U,\U'}(\tau)}$ denote the indicator of whether $\phi_{x_{<t}}^{\U,\U'} > \tau$. We have that \begin{align}
	\E[\U,\U']*{\Psi_{x_{<t}}^{\U,\U'}\cdot \phi_{x_{<t}}^{\U,\U'}} &= \E[\U,\U']*{\Psi_{x_{<t}}^{\U,\U'}\cdot \phi_{x_{<t}}^{\U,\U'} \cdot \left(\bone{\calE_{x_{<t}}^{\U,\U'}(\tau)} + \bone{\calE_{x_{<t}}^{\U,\U'}(\tau)^c}\right)} \\
	&\le \epsilon^2\cdot \E[\U,\U']*{\Psi_{x_{<t}}^{\U,\U'}\cdot \bone{\calE_{x_{<t}}^{\U,\U'}(\tau)}} + \tau\cdot \E[\U,\U']*{\Psi_{x_{<t}}^{\U,\U'}\cdot \bone{\calE_{x_{<t}}^{\U,\U'}(\tau)^c}} \\
	&\le \epsilon^2\cdot \underbrace{\E[\U,\U']*{\Psi_{x_{<t}}^{\U,\U'}\cdot \bone{\calE_{x_{<t}}^{\U,\U'}(\tau)}}}_{\circled{B}_{x_{<t}}} + \tau\cdot \underbrace{\E[\U,\U']*{\Psi_{x_{<t}}^{\U,\U'}}}_{\circled{G}_{x_{<t}}},
\end{align} where in the second step we used Part~\ref{fact:absbound} of Fact~\ref{fact:basic}. Note that for any transcript $x_{<t}$, $\Delta(x_{<t})^2 = \E[\U,\U']{\Psi^{\U,\U'}_{x_{<t}}} = \circled{G}_{x_{<t}}$, so by this and the fact that the likelihood ratio between two distributions always integrates to 1, \begin{equation}
	\E[x_{<t}\sim \nullhypo{t-1}]*{\frac{1}{\Delta^{(t-1)}(x_{<t})}\cdot \circled{G}_{x_{<t}}} = \E[x_{<t}\sim \nullhypo{t-1}]{\Delta^{(t-1)}(x_{<t})} = 1.\label{eq:one}
\end{equation}

We conclude that \begin{align}
	Z_t &\le \epsilon^2 \cdot \E[x_{<t}\sim \nullhypo{t-1}]*{\frac{1}{\Delta^{(t-1)}(x_{<t})}\cdot \circled{B}_{x_{<t}}} + \tau \cdot \E[x_{<t}\sim \nullhypo{t-1}]*{\frac{1}{\Delta^{(t-1)}(x_{<t})}\cdot \circled{G}_{x_{<t}}} \\
	&\le \epsilon^2 \cdot (1 + O(\epsilon^2/d))^{t-1} \E[x_{<t}\sim \nullhypo{t-1}]*{\circled{B}_{x_{<t}}} + \tau,
\end{align} where the second step follows by Lemma~\ref{lem:denom} and \eqref{eq:one}. So the challenge is to show that $\tau$ is the dominant quantity above, for appropriately chosen $\tau$.

Pick $\tau = \epsilon^2/d^{4/3}$. To upper bound $\E[x_{<t}\sim \nullhypo{t-1}]{\circled{B}_{x_{<t}}}$, apply Cauchy-Schwarz to get \begin{align}
	\E[x_{<t}\sim \nullhypo{t-1}]*{\circled{B}_{x_{<t}}} &\le \E[x_{<t}\sim \nullhypo{t-1},\U,\U']*{\left(\Psi^{\U,\U'}_{x_{<t}}\right)^2}^{1/2}\cdot \Pr[x_{<t}\sim \nullhypo{t-1},\U,\U']*{\calE_{x_{<t}}^{\U,\U'}(\tau)}^{1/2} \\
	&\le \E[x_{<t}\sim \nullhypo{t-1},\U,\U']*{\left(\Psi^{\U,\U'}_{x_{<t}}\right)^2}^{1/2}\cdot \exp\left(-\Omega(d^{1/3})\right),
\end{align}
where the second step follows by Theorem~\ref{thm:rho_tail}.

This, together with Lemma~\ref{lem:psi_main}, says that $\E[x_{<t}\sim \nullhypo{t-1}]{\circled{B}_{x_{<t}}}$ is indeed negligible for $t = o(d^{4/3}/\epsilon^2)$. For such $t$, $Z_t = O(\epsilon^2/d^{4/3})$, so by Lemma~\ref{lem:rewrite}, $\KL{\althypo{N}}{\nullhypo{N}} = o(1)$ as desired.
The desired result follows from Fact~\ref{fact:basic-lowerbound}.
\end{proof}


\section{Haar Tail Bounds}
\label{subsec:tails}

In this section we complete the proof of the two key estimates, Theorems~\ref{thm:moments} and~\ref{thm:rho_tail}, which were crucial to our proof of Theorem~\ref{thm:main}.
The following concentration inequality is key to our analysis:

\begin{theorem}[\cite{meckes2013spectral}, Corollary 17, see also \cite{anderson2010introduction}, Corollary 4.4.28]
	Equip $M \triangleq U(d)^k$ with the $L_2$-sum of Hilbert-Schmidt metrics. If $F: M\to \R$ is $L$-Lipschitz, then for any $t > 0$: \begin{equation}
		\Pr[(\U_1,...,\U_k)\in M]{|F(\U_1,...,\U_k) - \E{F(\U_1,...,\U_k)}|\ge t} \le e^{-d t^2/12L^2},
	\end{equation} where $\U_1,...,\U_k$ are independent unitary matrices drawn from the Haar measure.
	\label{thm:deusexmachina}
\end{theorem}

\subsection{Proof of Theorem~\ref{thm:moments}}
\label{subsec:moments_proof}

For convenience, Theorem~\ref{thm:moments} is restated below:

\moments*

To get intuition for this, consider again the special case where $\calM$ is an orthogonal POVM given by an orthonormal basis of $\C^d$. Then $p$ is uniform over $[d]$ and \begin{equation}K^{\U,\U'}_{\calM} =\frac{\epsilon^2}{d}\sum^d_{i=1}(\delta(\U_i) + \delta(\U'_i))^2 \le \frac{2\epsilon^2}{d}\sum^d_{i=1}(\delta(\U_i)^2 + \delta(\U'_i)^2),
\label{eq:specialK}
\end{equation} where $\delta(\cdot)$ is defined in \eqref{eq:deltadef}. The following is a standard fact: \begin{fact}
	For random unit vector $v\in\S^{d-1}$, $\E{\delta(v)^2} = \frac{1}{d+1}$.\label{fact:randomunitvector}
\end{fact}
\noindent
While this follows immediately from moments of random unit vectors, for pedagogical purposes we will give a proof using Weingarten calculus, as it will be a crucial ingredient later on.
Recall that for every $q \in \N$, there exists a corresponding Weingarten function $\Wg(\cdot, d): \calS_q \to \R$~\cite{weingarten1978asymptotic,collins2003moments}.
In the special case of $q = 2$, the symmetric group $\calS_q$ consists of two elements $e, \tau^*$, namely, the identity and non-identity permutation, respectively, and we have that $\Wg(e,d) = \frac{1}{d^2 - 1}$ and $\Wg(\tau^*,d) = -\frac{1}{d(d^2-1)}$.
We then have:
\begin{lemma}[Degree-2 case of \cite{collins2003moments}, Lemma 4.3]
	Let $e,\tau^*$ denote the identity and non-identity permutation of $\Sym_2$ respectively. For $d\ge 2$ and any $\A,\B\in\C^{d\times d}$, we have that\footnote{Note that this looks different from the statement in \cite{collins2003moments} only because they work with normalized trace $\text{tr}(\cdot)\triangleq \frac{1}{d}\Tr(\cdot)$.}
	\begin{equation}
		\E[\U]{\Tr((\A\U^{\dagger}\B\U)^2)} = \sum_{\sigma,\tau\in \Sym_2}\iprod{\A}_{\sigma}\iprod{\B}_{\tau}\Wg(\sigma\tau^{-1},d) \; .
	\end{equation}
	\label{lem:weincalc}
\end{lemma}

\begin{proof}[Proof of Fact~\ref{fact:randomunitvector}]
	Let $\Proj\triangleq e_1e_1^{\dagger}$ and note that $\delta(v)$ is identical in distribution to the quantity $\Tr(\Proj\U^{\dagger}\X'\U)$. By Lemma~\ref{lem:weincalc}, \begin{equation}
		\E[v]{\delta(v)^2} = \E[\U]{\Tr(\Proj\U^{\dagger}\X'\U)^2} = \sum_{\sigma,\tau\in\Sym_2}\iprod{\Proj}_{\sigma}\iprod{\X'}_{\tau}\Wg(\sigma\tau^{-1},d).
	\end{equation}  Note that $\iprod{\X'}_{\tau} = d\cdot\bone{\tau = \tau^*}$ and $\iprod{\Proj}_{\sigma} = 1$ for all $\sigma\in\Sym_2$, so \begin{equation}
		\E[v]{\delta(v)^2} = d\left(\frac{1}{d^2 - 1} - \frac{1}{d(d^2-1)}\right) = \frac{1}{d+1}
	\end{equation} as claimed.
\end{proof}

Furthermore, it is known that $\delta(v)^2$ concentrates around its expectation. So if the columns of $\U$ were actually \emph{independent} random unit vectors, we would conclude that $K^{\U,\U'}_{\calM} = O(\epsilon^2/d)$ with high probability and obtain \eqref{eq:Kdef} for the special case where $\calM$ is orthogonal.

To circumvent the issue of dependence among the columns of Haar-random $\U$, we will invoke Theorem~\ref{thm:deusexmachina}. The following is a toy version of the more general result that we show later in our proof of Theorem~\ref{thm:moments} (see Lemma~\ref{lem:Ztail}):

\begin{lemma}
	For any $t > 0$, $\Pr[\U\sim\calD]*{\left(\sum^d_{i=1}\delta(\U_i)^2\right)^{1/2} \ge 1 + t} \le \exp\left(-\Omega(dt^2)\right)$.
	\label{lem:basicZtail}
\end{lemma}

\begin{proof}
	By Jensen's and Fact~\ref{fact:randomunitvector}, \begin{equation}
		\E*{\left(\sum^d_{i=1}\delta(\U_i)^2\right)^{1/2}} \le \E*{\sum^d_{i=1}\delta(\U_i)^2}^{1/2} = \left(\frac{d}{d+1}\right)^{1/2}\le 1.
	\end{equation} We wish to invoke Theorem~\ref{thm:deusexmachina}, so it suffices to show that $G:\U\mapsto (\sum^d_{i=1}\delta(\U_i)^2)^{1/2}$ is $O(1)$-Lipschitz. Recalling the definition of $\X'$ from Construction~\ref{constr:quantum}, note that \begin{equation}\left(\sum^d_{i=1}\delta(\U_i)^2\right)^{1/2} = \norm{\diag(\U^{\dagger}\X'\U)}_{HS}.
	\label{eq:deltaHS}
	\end{equation} Take any $\U,\V\in U(d)$ and note \begin{align}
		G(\U) - G(\V) &\le \sqrt{\sum^d_{i=1}\abs{(\U^{\dagger}\X'\U)_{ii} - (\V^{\dagger}\X'\V)_{ii}}^2} \\
		&\le \norm{\U^{\dagger}\X'\U - \V^{\dagger}\X'\V}_{HS} \\
		&= \norm{\U^{\dagger}\X'(\U - \V) + (\V - \U)^{\dagger}\X'\V}_{HS} \\
		&\le 2\norm{\X'}_2\norm{\U - \V}_{HS} = 2\norm{\U - \V}_{HS},
	\end{align} where the first step follows by Cauchy-Schwarz. So $G(\U)$ is 2-Lipschitz as desired.
\end{proof}

Eq.~\eqref{eq:specialK}, Fact~\ref{fact:randomunitvector}, and Lemma~\ref{lem:basicZtail}, together with integration by parts, allow us to conclude Theorem~\ref{thm:moments} in the special case where $\calM$ is orthogonal. Guided by the arguments above, we now proceed to our actual proof of Theorem~\ref{thm:moments}.

\begin{proof}[Proof of Theorem~\ref{thm:moments}]

Let $\calM$ be an arbitrary POVM. We first show a bound on $\E[\U,\U']{K^{\U,\U'}_{\calM}}$, generalizing Fact~\ref{fact:randomunitvector}:

\begin{lemma}
	$\E[\U,\U']{K^{\U,\U'}_{\calM}} \le \frac{\epsilon^2}{d+1}$.
	\label{lem:secondmoment}
\end{lemma}

\begin{proof}
	Note that $K^{\U,\U'}_{\calM} = 2\E[x\sim p]*{g^{\U}_{\calM}(x)^2} + 2\phi^{\U,\U'}_{\calM}$. We will suppress the subscripts for the rest of this proof. Clearly we have that $\E[\U,\U']{\phi^{\U,\U'}} = 0$, so it remains to bound $\E[x\sim p,\U,\U']*{g^{\U}(x)^2}$. Let $\tau^*\in\Sym_2$ denote the non-identity permutation. For any fixed $x$, by Lemma~\ref{lem:weincalc}, \begin{align}
		\E[\U]{g^{\U}(x)^2} &= \sum_{\sigma,\tau\in\Sym_2}\iprod{\X}_{\tau}\iprod{\hatM_x}_{\sigma}\Wg(\sigma\tau^{-1},d) \\
		&= \iprod{\X}_{\tau^*} \left(\Tr(\hatM_x^2)\cdot \Wg(e,d) + \Tr(\hatM_x)^2\cdot \Wg(\tau^*,d)\right) \\
		&= \epsilon^2\cdot d\cdot\left(\frac{1}{d^2 - 1}\Tr(\hatM_x^2) - \frac{1}{d(d^2-1)}\right) \le \frac{\epsilon^2}{d+1},\label{eq:gzsingle}
	\end{align} where the second step follows by the fact that $\iprod{\X}_{\tau} = \epsilon^2\cdot d\cdot \bone{\tau= \tau^*}$, and the last step follows by the fact that $\Tr(\hatM_x^2) \le 1$. As \eqref{eq:gzsingle} holds for any outcome $x$, $\E[x\sim p,\U,\U']{g^{\U}(x)^2} \le \frac{\epsilon^2}{d+1}$ as desired.
\end{proof}

We next show the following tail bound generalizing Lemma~\ref{lem:basicZtail}:

\begin{lemma}
	There are absolute constants $c,c'>0$ such that 
	\begin{equation}
		\Pr[\U,\U'\sim\calD]*{K^{\U,\U'}_{\calM} > c\epsilon^2/d + t} \le \exp(-c' t d^2/\epsilon^2))
	\end{equation} for any $t > c \epsilon^2/d$.
	\label{lem:Ztail}
\end{lemma}

\begin{proof}
	We wish to apply Theorem~\ref{thm:deusexmachina}. We will show that $F: (\U,\U')\mapsto \left(K^{\U,\U'}_{\calM}\right)^{1/2}$ is $L$-Lipschitz for $L = O(\epsilon/\sqrt{d})$. As $\E{F(\U,\U')} \le \E[\U,\U']*{K^{\U,\U'}_{\calM}}^{1/2} \le \frac{\epsilon}{\sqrt{d+1}}$ by Jensen's and Lemma~\ref{lem:secondmoment}, this would imply that for any $s > 0$, \begin{equation}
		\Pr[\U\sim\calD]*{F(\U,\U') > \frac{\epsilon}{\sqrt{d+1}} + s} \le e^{-d^2s^2/12\epsilon^2},
	\end{equation} from which the lemma would follow by taking $s = \sqrt{t}$.

	To show Lipschitz-ness, note that \begin{equation}
		F(\U,\U') = \E[x\sim p]*{\left(g^{\U}_{\calM}(x) + g^{\U'}_{\calM}(x)\right)^2} \le \E[x\sim p]*{g^{\U}_{\calM}(x)^2}^{1/2} + \E[x\sim p]*{g^{\U'}_{\calM}(x)^2}^{1/2},
	\end{equation} so the proof is complete given Lemma~\ref{lem:Glipschitz} below.
\end{proof}

\begin{lemma}
	The function $G: \U\mapsto \E[x\sim p]*{g^{\U}_{\calM}(x)^2}^{1/2}$ is $O(\epsilon/\sqrt{d})$-Lipschitz.
	\label{lem:Glipschitz}
\end{lemma} 

\begin{proof}
	Take any $\U,\V\in U(d)$ and note that by triangle inequality, \begin{equation}
		G(U) - G(V) \le \E[x\sim p]*{\left(g^{\U}_{\calM}(x) - g^{\V}_{\calM}(x)\right)^2}^{1/2},
	\end{equation} so it suffices to show \begin{equation}
		\E[x\sim p]*{\left(g^{\U}_{\calM}(x) - g^{\V}_{\calM}(x)\right)^2} \le O(\epsilon^2/d)\cdot\norm{\U - \U'}^2_{HS}.
		\label{eq:wantlipschitz}
	\end{equation}
	$\A\triangleq \U^{\dagger}\X\U - \V^{\dagger}\X\V$ is Hermitian, so write its eigendecomposition $\A = \W^{\dagger}\Sig\W$. Then \begin{equation}
		g^{\U}_{\calM}(x) - g^{\V}_{\calM}(x) = \iprod*{\hatM_x, \A} = \iprod*{\W \hatM_x \W^{\dagger},\Sig} = \iprod*{\diag(\W\hatM_x\W^{\dagger}),\Sig},
	\end{equation} so we may assume without loss of generality that $\A$ and $\hatM_x$ are diagonal, in which case by Jensen's, \begin{equation}
		\left(g^{\U}_{\calM}(x) - g^{\V}_{\calM}(x)\right)^2 = \iprod{\hatM_x,\A}^2 = \left(\sum^d_{i=1} (\hatM_x)_{ii}\A_{ii}\right)^2 \le \sum^d_{i=1} (\hatM_x)_{ii}\A^2_{ii}.
	\end{equation}
	Recalling the definition of $p$ and letting $\mu$ denote the measure over $\Omega(\calM)$ associated to $\calM$ (see Definition~\ref{defn:generic}), we see that the left-hand side of \eqref{eq:wantlipschitz} becomes \begin{multline}
		\frac{1}{d}\int_{\Omega(\calM)} \Tr(M_x)\cdot \iprod{\hatM_x,\A}^2 \, \d\mu = \frac{1}{d}\int_{\Omega(\calM)} \sum_{i\in[d]} \Tr(M_x)\cdot (\hatM_x)_{ii}\A^2_{ii} \, \d \mu = \frac{1}{d}\int_{\Omega(\calM)} \sum_{i\in[d]} (M_x)_{ii}\A^2_{ii}\\
		= \frac{1}{d}\norm{\A}^2_{HS} = \frac{1}{d}\norm{\U^{\dagger}\X(\U - \V) + (\V - \U)^{\dagger}\X\V}^2_{HS} \le \frac{2\epsilon^2}{d}\norm{\U - \V}^2_{HS},
	\end{multline} where the third step follows from the fact that $\int_{\Omega(\calM)} M_x\, \d\mu = \Id$, completing the proof of \eqref{eq:wantlipschitz}.
\end{proof}

To complete the proof of Theorem~\ref{thm:moments}, we would like to apply Fact~\ref{fact:stieltjes} to the random variable $Z\triangleq 1 + \gamma\cdot K^{\U,\U'}_{\calM}$ and the function $f(Z)\triangleq Z^n$. Note that this random variable is nonnegative and upper bounded by $1+C\cdot\gamma\cdot \epsilon^2$ for some absolute constant $C > 0$. So \begin{align}
	\E[\U,\U']*{\left(1 + \gamma K^{\U,\U'}_{\calM}\right)^n} &\le 2(1 + O(\gamma\epsilon^2/d))^n + \int^{1 + C\gamma\epsilon^2}_{1 + c\gamma\epsilon^2/d}nt^{n-1}\cdot e^{-\Omega(t\cdot d^2/\epsilon^2)}\ \d t \\
	&\le 2(1 + O(\gamma\epsilon^2/d))^n + \int^{\infty}_0 nt^{n-1}\cdot e^{-\Omega(t\cdot d^2/\epsilon^2)}\ \d t \\
	&= 2(1 + O(\gamma\epsilon^2/d))^n + n!\cdot O(\epsilon^2/d^2)^n \\
	&\le e^{O(\gamma\epsilon^2 n/d)},
\end{align} where in the last step we used that $n!\cdot O(\epsilon^2/d^2)^n$ is negligible when $n = o(d^2/\epsilon^2)$
\end{proof}

\subsection{Proof of Theorem~\ref{thm:rho_tail}}

For convenience, Theorem~\ref{thm:rho_tail} is restated below. Recall from the discussion in Section~\ref{sec:intuition} that this can be thought of as the ``quantum analogue'' of binomial tail bounds:

\rhotail*

\begin{proof}[Proof of Theorem~\ref{thm:rho_tail}]
	Define $G$ as in Lemma~\ref{lem:Glipschitz}. Fix any $\U'$ and consider the function $F_{\U'}: \U\mapsto \phi^{\U,\U'}_{\calM}$. 
	First note that \begin{equation}
		\E[\U]{F_{\U'}(\U)} = \E[x\sim p]{g^{\U'}_{\calM}(x)\cdot \E[\U]{g^{\U}_{\calM}(x)}}= 0
	\end{equation} by Part~\ref{fact:zero} of Fact~\ref{fact:basic}.
	Next, note that by Cauchy-Schwarz, \begin{equation}
		F_{\U'}(\U) - F_{\U'}(\V) \le \E[x\sim p]*{\left(g^{\U}_{\calM}(x) - g^{\V}_{\calM}(x)\right)^2}^{1/2} \cdot G(\U'),
	\end{equation}
	which by \eqref{eq:wantlipschitz} is $O(\epsilon/\sqrt{d})\cdot G(\U')$-Lipschitz.
	So for any fixed $\U'$, Theorem~\ref{thm:deusexmachina} implies \begin{equation}
		\Pr[\U]{\abs{F_{\U'}(\U)} > t} \le \exp\left(-C\cdot \frac{d^2t^2}{\epsilon^2 G(\U')^2}\right)\label{eq:integrand}
	\end{equation} for some absolute constant $C> 0$. We would like to integrate over $\U'$ to get a tail bound for $\phi^{\U,\U'}_{\calM}$ as a function of both $\U$ and $\U'$.

	To this end, we can apply Fact~\ref{fact:stieltjes} to the random variable $Y\triangleq G(\U') \in [0,\epsilon]$. Recall from \eqref{eq:gzsingle} and Jensen's that $\E{Y}\le\epsilon/\sqrt{d+1}$. Furthermore, by Lemma~\ref{lem:Glipschitz} and Theorem~\ref{thm:deusexmachina}, there is an absolute constant $C'>0$ such that \begin{equation}
		\Pr{Y > \epsilon/\sqrt{d+1} + s} \le \exp\left(-C'd^2s^2/\epsilon^2\right).\label{eq:Yconc}
	\end{equation}
	So we can take the parameters in Fact~\ref{fact:stieltjes} as follows: set $a \triangleq 2\epsilon/\sqrt{d+1}$, tail bound function $\tau(x)\triangleq \exp\left(-C'\cdot \frac{d^2}{\epsilon^2}\left(x - \frac{\epsilon}{\sqrt{d+1}}\right)^2\right)$ for absolute constant $C'>0$, and $f(Y)\triangleq \exp\left(-C\cdot\frac{d^2t^2}{\epsilon^2 Y^2}\right)$. By \eqref{eq:integrand}, $\Pr[\U,\U']*{\abs*{\phi^{\U,\U'}_{\calM}}>t} \le \E{f(Y)}$, and by Fact~\ref{fact:stieltjes}, \begin{equation}
		\E{f(Y)} \le 2\exp\left(-\Omega(d^3t^2/\epsilon^4)\right) + \int^{\epsilon}_{2\epsilon/\sqrt{d+1}} \frac{2C d^2t^2}{\epsilon^2 x^3} \exp\left(-\frac{d^2}{\epsilon^2} \left(\frac{Ct^2}{x^2} + C'\left(x - \frac{\epsilon}{\sqrt{d+1}}\right)^2\right)\right) \ \d x.
	\end{equation} Note that by AM-GM, for $x\ge 2\epsilon/\sqrt{d+1}$ we have that \begin{equation}
		\frac{Ct^2}{x^2} + C'\left(x - \frac{\epsilon}{\sqrt{d+1}}\right)^2 \ge 2t\cdot (C\cdot C')^{1/2}\cdot \left(1 - \frac{\epsilon/\sqrt{d+1}}{x}\right) \ge t\cdot (C\cdot C')^{1/2}.
	\end{equation} We conclude that \begin{equation}
		\Pr[\U,\U']*{\abs*{\phi^{\U,\U'}_{\calM}}>t} \le 2\exp\left(-\Omega(d^3t^2/\epsilon^4)\right) + \frac{2Cd^{7/2}t^2}{\epsilon^4}\cdot \exp\left(-\Omega\left(d^2t/\epsilon^2\right)\right).
	\end{equation} In particular, for $t \ge \Omega(\epsilon^2/d^{1.99})$, we have that \begin{equation}
		\Pr[\U,\U']*{\abs*{\phi^{\U,\U'}_{\calM}}>t} \ge \exp\left(-\Omega\left(\Min{\frac{d^3t^2}{\epsilon^4}}{\frac{d^2t}{\epsilon^2}}\right)\right)
	\end{equation} as claimed.
\end{proof}

\paragraph{Acknowledgments}

The authors would like to thank Ofer Zeitouni for pointing out the existence of Theorem~\ref{thm:deusexmachina}, and Robin Kothari and Jeongwan Haah for helpful preliminary discussions about quantum tomography.

\bibliographystyle{alpha}
\bibliography{biblio}

\newcommand{\etalchar}[1]{$^{#1}$}
\begin{thebibliography}{ACBFS02}

\bibitem[ACBFS02]{auer2002nonstochastic}
Peter Auer, Nicolo Cesa-Bianchi, Yoav Freund, and Robert~E Schapire.
\newblock The nonstochastic multiarmed bandit problem.
\newblock {\em SIAM journal on computing}, 32(1):48--77, 2002.

\bibitem[ACK14]{acharya2014chasm}
Jayadev Acharya, Cl{\'e}ment~L Canonne, and Gautam Kamath.
\newblock A chasm between identity and equivalence testing with conditional
  queries.
\newblock {\em arXiv preprint arXiv:1411.7346}, 2014.

\bibitem[AGKE15]{aolita2015reliable}
Leandro Aolita, Christian Gogolin, Martin Kliesch, and Jens Eisert.
\newblock Reliable quantum certification of photonic state preparations.
\newblock {\em Nature communications}, 6(1):1--8, 2015.

\bibitem[AGZ10]{anderson2010introduction}
Greg~W Anderson, Alice Guionnet, and Ofer Zeitouni.
\newblock {\em An introduction to random matrices}, volume 118.
\newblock Cambridge university press, 2010.

\bibitem[AISW19]{acharya2019measuring}
Jayadev Acharya, Ibrahim Issa, Nirmal~V Shende, and Aaron~B Wagner.
\newblock Measuring quantum entropy.
\newblock In {\em 2019 IEEE International Symposium on Information Theory
  (ISIT)}, pages 3012--3016. IEEE, 2019.

\bibitem[ANSV08]{audenaert2008asymptotic}
Koenraad~MR Audenaert, Michael Nussbaum, Arleta Szko{\l}a, and Frank
  Verstraete.
\newblock Asymptotic error rates in quantum hypothesis testing.
\newblock {\em Communications in Mathematical Physics}, 279(1):251--283, 2008.

\bibitem[AOK09]{al2009statistics}
Vladimir Al~Osipov and Eugene Kanzieper.
\newblock Statistics of thermal to shot noise crossover in chaotic cavities.
\newblock {\em Journal of Physics A: Mathematical and Theoretical},
  42(47):475101, 2009.

\bibitem[BB96]{brouwer1996diagrammatic}
PW~Brouwer and CWJ Beenakker.
\newblock Diagrammatic method of integration over the unitary group, with
  applications to quantum transport in mesoscopic systems.
\newblock {\em Journal of Mathematical Physics}, 37(10):4904--4934, 1996.

\bibitem[BB00]{blanter2000shot}
Ya~M Blanter and Markus B{\"u}ttiker.
\newblock Shot noise in mesoscopic conductors.
\newblock {\em Physics reports}, 336(1-2):1--166, 2000.

\bibitem[BB16]{belovs2016polynomial}
Aleksandrs Belovs and Eric Blais.
\newblock A polynomial lower bound for testing monotonicity.
\newblock In {\em Proceedings of the forty-eighth annual ACM symposium on
  Theory of Computing}, pages 1021--1032, 2016.

\bibitem[BC09]{barnett2009quantum}
Stephen~M Barnett and Sarah Croke.
\newblock Quantum state discrimination.
\newblock {\em Advances in Optics and Photonics}, 1(2):238--278, 2009.

\bibitem[BC18]{bhattacharyya2018property}
Rishiraj Bhattacharyya and Sourav Chakraborty.
\newblock Property testing of joint distributions using conditional samples.
\newblock {\em ACM Transactions on Computation Theory (TOCT)}, 10(4):1--20,
  2018.

\bibitem[BCB12]{bubeck2012regret}
S{\'e}bastien Bubeck and Nicolo Cesa-Bianchi.
\newblock Regret analysis of stochastic and nonstochastic multi-armed bandit
  problems.
\newblock {\em Foundations and Trends{\textregistered} in Machine Learning},
  5(1):1--122, 2012.

\bibitem[BCP{\etalchar{+}}17]{baleshzar2017optimal}
Roksana Baleshzar, Deeparnab Chakrabarty, Ramesh Krishnan~S Pallavoor, Sofya
  Raskhodnikova, and C~Seshadhri.
\newblock Optimal unateness testers for real-valued functions: Adaptivity
  helps.
\newblock {\em arXiv preprint arXiv:1703.05199}, 2017.

\bibitem[Bee97]{beenakker1997random}
Carlo~WJ Beenakker.
\newblock Random-matrix theory of quantum transport.
\newblock {\em Reviews of modern physics}, 69(3):731, 1997.

\bibitem[Bel18]{belovs2018adaptive}
Aleksandrs Belovs.
\newblock Adaptive lower bound for testing monotonicity on the line.
\newblock {\em arXiv preprint arXiv:1801.08709}, 2018.

\bibitem[BOW19]{buadescu2019quantum}
Costin B{\u{a}}descu, Ryan O'Donnell, and John Wright.
\newblock Quantum state certification.
\newblock In {\em Proceedings of the 51st Annual ACM SIGACT Symposium on Theory
  of Computing}, pages 503--514, 2019.

\bibitem[Can17]{canonne2017survey}
Cl{\'e}ment~L Canonne.
\newblock A survey on distribution testing.
\newblock 2017.

\bibitem[CDGR18]{canonne2018testing}
Cl{\'e}ment~L Canonne, Ilias Diakonikolas, Themis Gouleakis, and Ronitt
  Rubinfeld.
\newblock Testing shape restrictions of discrete distributions.
\newblock {\em Theory of Computing Systems}, 62(1):4--62, 2018.

\bibitem[CDS10]{chiribella2010barycentric}
Giulio Chiribella, Giacomo~Mauro D’Ariano, and Dirk Schlingemann.
\newblock Barycentric decomposition of quantum measurements in finite
  dimensions.
\newblock {\em Journal of mathematical physics}, 51(2):022111, 2010.

\bibitem[CDVV14]{chan2014optimal}
Siu-On Chan, Ilias Diakonikolas, Paul Valiant, and Gregory Valiant.
\newblock Optimal algorithms for testing closeness of discrete distributions.
\newblock In {\em Proceedings of the twenty-fifth annual ACM-SIAM symposium on
  Discrete algorithms}, pages 1193--1203. SIAM, 2014.

\bibitem[CFGM16]{chakraborty2016power}
Sourav Chakraborty, Eldar Fischer, Yonatan Goldhirsh, and Arie Matsliah.
\newblock On the power of conditional samples in distribution testing.
\newblock {\em SIAM Journal on Computing}, 45(4):1261--1296, 2016.

\bibitem[Che00]{chefles2000quantum}
Anthony Chefles.
\newblock Quantum state discrimination.
\newblock {\em Contemporary Physics}, 41(6):401--424, 2000.

\bibitem[Col03]{collins2003moments}
Beno{\^\i}t Collins.
\newblock Moments and cumulants of polynomial random variables on
  unitarygroups, the itzykson-zuber integral, and free probability.
\newblock {\em International Mathematics Research Notices}, 2003(17):953--982,
  2003.

\bibitem[CRS14]{canonne2014testing}
Cl{\'e}ment Canonne, Dana Ron, and Rocco~A Servedio.
\newblock Testing equivalence between distributions using conditional samples.
\newblock In {\em Proceedings of the twenty-fifth annual ACM-SIAM symposium on
  Discrete algorithms}, pages 1174--1192. SIAM, 2014.

\bibitem[CRS15]{canonne2015testing}
Cl{\'e}ment~L Canonne, Dana Ron, and Rocco~A Servedio.
\newblock Testing probability distributions using conditional samples.
\newblock {\em SIAM Journal on Computing}, 44(3):540--616, 2015.

\bibitem[CW20]{cotler2020quantum}
Jordan Cotler and Frank Wilczek.
\newblock Quantum overlapping tomography.
\newblock {\em Physical Review Letters}, 124(10):100401, 2020.

\bibitem[CWX17a]{chen2017beyond}
Xi~Chen, Erik Waingarten, and Jinyu Xie.
\newblock Beyond talagrand functions: new lower bounds for testing monotonicity
  and unateness.
\newblock In {\em Proceedings of the 49th Annual ACM SIGACT Symposium on Theory
  of Computing}, pages 523--536, 2017.

\bibitem[CWX17b]{chen2017boolean}
Xi~Chen, Erik Waingarten, and Jinyu Xie.
\newblock Boolean unateness testing with $\tilde{O} (n^{3/4})$ adaptive
  queries.
\newblock In {\em 2017 IEEE 58th Annual Symposium on Foundations of Computer
  Science (FOCS)}, pages 868--879. IEEE, 2017.

\bibitem[DKN14]{diakonikolas2014testing}
Ilias Diakonikolas, Daniel~M Kane, and Vladimir Nikishkin.
\newblock Testing identity of structured distributions.
\newblock In {\em Proceedings of the twenty-sixth annual ACM-SIAM symposium on
  Discrete algorithms}, pages 1841--1854. SIAM, 2014.

\bibitem[dSLCP11]{da2011practical}
Marcus~P da~Silva, Olivier Landon-Cardinal, and David Poulin.
\newblock Practical characterization of quantum devices without tomography.
\newblock {\em Physical Review Letters}, 107(21):210404, 2011.

\bibitem[FGLE12]{flammia2012quantum}
Steven~T Flammia, David Gross, Yi-Kai Liu, and Jens Eisert.
\newblock Quantum tomography via compressed sensing: error bounds, sample
  complexity and efficient estimators.
\newblock {\em New Journal of Physics}, 14(9):095022, 2012.

\bibitem[FL11]{flammia2011direct}
Steven~T Flammia and Yi-Kai Liu.
\newblock Direct fidelity estimation from few pauli measurements.
\newblock {\em Physical review letters}, 106(23):230501, 2011.

\bibitem[Gol17]{goldreich2017introduction}
Oded Goldreich.
\newblock {\em Introduction to property testing}.
\newblock Cambridge University Press, 2017.

\bibitem[HHJ{\etalchar{+}}17]{haah2017sample}
Jeongwan Haah, Aram~W Harrow, Zhengfeng Ji, Xiaodi Wu, and Nengkun Yu.
\newblock Sample-optimal tomography of quantum states.
\newblock {\em IEEE Transactions on Information Theory}, 63(9):5628--5641,
  2017.

\bibitem[KS16]{khot2016n}
Subhash Khot and Igor Shinkar.
\newblock An $\tilde{o} (n)$ queries adaptive tester for unateness.
\newblock In {\em Approximation, Randomization, and Combinatorial Optimization.
  Algorithms and Techniques (APPROX/RANDOM 2016)}. Schloss
  Dagstuhl-Leibniz-Zentrum fuer Informatik, 2016.

\bibitem[KSS09]{khoruzhenko2009systematic}
BA~Khoruzhenko, DV~Savin, and H-J Sommers.
\newblock Systematic approach to statistics of conductance and shot-noise in
  chaotic cavities.
\newblock {\em Physical Review B}, 80(12):125301, 2009.

\bibitem[KT19]{kamath2019anaconda}
Gautam Kamath and Christos Tzamos.
\newblock Anaconda: A non-adaptive conditional sampling algorithm for
  distribution testing.
\newblock In {\em Proceedings of the Thirtieth Annual ACM-SIAM Symposium on
  Discrete Algorithms}, pages 679--693. SIAM, 2019.

\bibitem[LeC73]{lecam1973convergence}
Lucien LeCam.
\newblock Convergence of estimates under dimensionality restrictions.
\newblock {\em The Annals of Statistics}, 1(1):38--53, 1973.

\bibitem[Md16]{montanaro2016survey}
A~Montanaro and RM~deWolf.
\newblock A survey of quantum property testing.
\newblock {\em Theory of Computing}, (Graduate Surveys), 2016.

\bibitem[MHC13]{moran2013positive}
Bill Moran, Stephen Howard, and Doug Cochran.
\newblock Positive-operator-valued measures: a general setting for frames.
\newblock In {\em Excursions in Harmonic Analysis, Volume 2}, pages 49--64.
  Springer, 2013.

\bibitem[MM13]{meckes2013spectral}
Elizabeth Meckes and Mark Meckes.
\newblock Spectral measures of powers of random matrices.
\newblock {\em Electronic communications in probability}, 18, 2013.

\bibitem[OW15]{o2015quantum}
Ryan O'Donnell and John Wright.
\newblock Quantum spectrum testing.
\newblock In {\em Proceedings of the forty-seventh annual ACM symposium on
  Theory of computing}, pages 529--538, 2015.

\bibitem[OW16]{o2016efficient}
Ryan O'Donnell and John Wright.
\newblock Efficient quantum tomography.
\newblock In {\em Proceedings of the forty-eighth annual ACM symposium on
  Theory of Computing}, pages 899--912, 2016.

\bibitem[OW17]{o2017efficient}
Ryan O'Donnell and John Wright.
\newblock Efficient quantum tomography ii.
\newblock In {\em Proceedings of the 49th Annual ACM SIGACT Symposium on Theory
  of Computing}, pages 962--974, 2017.

\bibitem[Pan08]{paninski2008coincidence}
Liam Paninski.
\newblock A coincidence-based test for uniformity given very sparsely sampled
  discrete data.
\newblock {\em IEEE Transactions on Information Theory}, 54(10):4750--4755,
  2008.

\bibitem[Wei78]{weingarten1978asymptotic}
Don Weingarten.
\newblock Asymptotic behavior of group integrals in the limit of infinite rank.
\newblock {\em Journal of Mathematical Physics}, 19(5):999--1001, 1978.

\bibitem[Wri16]{wright2016learn}
John Wright.
\newblock {\em How to learn a quantum state}.
\newblock PhD thesis, Carnegie Mellon University Pittsburgh, PA, 2016.

\end{thebibliography}

\appendix


\section{A Simple Tester Using \texorpdfstring{$O(d^{3/2}/\epsilon^2)$}{d32/eps2} Unentangled Measurements}
\label{subsec:tester}

In this section we prove the upper bound part of Theorem~\ref{thm:nonadaptive}, i.e.

\begin{theorem}
	Given $0 < \epsilon<1$ and copy access to $\rho$, there is an algorithm \textsc{TestMixed}($\rho,d,\epsilon)$ that makes unentangled measurements on $N = O(d^{3/2}/\epsilon^2)$ copies of $\rho$ and with probability $4/5$ distinguishes whether $\norm{\rho - \mm}_1\ge \epsilon$ or $\rho = \mm$.
	\label{thm:tester}
\end{theorem}

Our mixedness tester is extremely simple: pick a random orthogonal POVM corresponding to a Haar-random basis of $\C^d$, measure $O(d^{3/2}/\epsilon^2)$ copies of $\rho$ with this POVM, and use these measurement outcomes to check whether the distribution over measurement outcomes is too far (in $L_2$ distance) from uniform, in which case $\rho$ is far from $\mm$.

We will need the following result on classical uniformity testing in $L_2$.
We note that this result first appeared in~\cite{chan2014optimal,diakonikolas2014testing} with slightly incorrect proofs, which were fixed in~\cite{canonne2018testing}.

\begin{theorem}[\cite{chan2014optimal,diakonikolas2014testing,canonne2018testing}]
	Given $0 < \epsilon<1$ and sample access to a distribution $q$ over $[d]$, there is an algorithm \textsc{TestUniformityL2}($q,d,\epsilon)$ that uses $N = O(\sqrt{d}/\epsilon^2)$ samples from $q$ and with probability $9/10$ distinguishes whether $q$ is the uniform distribution over $[d]$ or $\epsilon/\sqrt{d}$-far in $L_2$ distance from the uniform distribution.
	\label{thm:ilias}
\end{theorem}

Certainly when $\rho = \mm$, for any orthogonal POVM corresponding to an orthonormal basis of $\C^d$, the induced distribution over measurement outcomes will be the uniform distribution over $d$ elements. The point is that when $\rho$ is $\epsilon$-far from maximally mixed, for a Haar-random orthogonal POVM, the induced distribution over measurement outcomes will be $O(\epsilon/\sqrt{d})$-far in $L_2$ distance from the uniform distribution over $d$ elements with high probability. So Theorem~\ref{thm:ilias} would imply an algorithm for testing mixedness which makes $O(d^{3/2}/\epsilon^2)$ unentangled, nonadaptive measurements. Formally, our algorithm is specified in Algorithm~\ref{alg:test_mixed} below.

\begin{algorithm}[h]\caption{\textsc{TestMixed}}
\label{alg:test_mixed}
\begin{algorithmic}[1]
	\State \textbf{Input}: $N \triangleq \Theta(d^{3/2}/\epsilon^2)$ copies of unknown state $\rho$
	\State \textbf{Output}: $\mathsf{NO}$ if $\norm{\rho - \frac{1}{d}\cdot\Id}_{1} \ge \epsilon$, $\mathsf{YES}$ if $\rho = \frac{1}{d}\cdot\Id$, with probability $4/5$.
			\State Sample a Haar-random unitary matrix $U\in\C^{d\times d}$.
			\State Define the POVM $\{\ket{U_i}\bra{U_i}\}_{1\le i\le d}$ and measure with this POVM $N$ times to get $N$ independent samples from the distribution $q$ over measurement outcomes.\label{step:POVM}
			\State Run \textsc{TestUniformityL2}($q,d,\epsilon/\sqrt{d}$) from Theorem~\ref{thm:ilias} to test uniformity of $q$.
			\If{$q$ far from uniform}
				\State Output $\mathsf{NO}$.
			\Else
				\State Output $\mathsf{YES}$.
			\EndIf 
\end{algorithmic}
\end{algorithm}

\begin{proof}
As the POVM defined in Step~\ref{step:POVM} of \textsc{TestMixed} is Haar-random, we may assume $\rho = \Lam$ without loss of generality, where $\Lam$ is the diagonal matrix whose first $d/2$ diagonal entries are $1/d + \epsilon/d$ and whose last $d/2$ diagonal entries are $1/d-\epsilon/d$. Also define the diagonal matrix $\X'$ whose first $d/2$ diagonal entries are $1$ and whose last $d/2$ diagonal entries are $-1$.

Let $q$ be the distribution over measurement outcomes, and let $u$ be the uniform distribution over $[d]$. Note that for any $i\in[d]$, the marginal probability $q_i = (\U^{\dagger}\Lam\U)_{ii}$, so \begin{equation}
	\norm{q - u}_2 = \left(\sum^d_{i=1}(q_i - 1/d)^2\right)^{1/2} = \norm{\diag(\U^{\dagger}\overline{\Lam}\U)}_{HS} = \frac{\epsilon}{d} \norm{\diag(\U^{\dagger}\X'\U)}_{HS}.\label{eq:trdiagsq}
\end{equation}

Recall from \eqref{eq:deltaHS} and the discussion at the beginning of Section~\ref{subsec:moments_proof} that we have a fine understanding of the expectation (Fact~\ref{fact:randomunitvector}) and tails (Lemma~\ref{lem:basicZtail}) of this quantity. Indeed, by \eqref{eq:deltaHS} and Lemma~\ref{lem:basicZtail}, we know that $\Pr[\U\sim\calD]{\norm{\diag(\U^{\dagger}\X'\U)}_{HS} \ge 1 + t} \le \exp\left(-\Omega(dt^2)\right)$. We conclude from this and \eqref{eq:trdiagsq} that $\norm{q - u}_2 \le 2\epsilon/d$ with probability $\exp(-\Omega(d))$, from which the theorem follows by the guarantees of \textsc{TestUniformityL2} in Theorem~\ref{thm:ilias}.
\end{proof}


\section{Chain Rule Proof of Theorem~\ref{thm:paninski}}
\label{app:paninski}

Here we give a proof of the tight $\Omega(\sqrt{d}/\epsilon^2)$ bound from Theorem~\ref{thm:paninski} using the chain rule. The only part of the proof of the weaker Theorem~\ref{thm:paninski_weak} that we need here is Lemma~\ref{lem:rewrite_paninski}, which we restate here for convenience.

\rewritepan*

Define $\epsilon'\triangleq \log\frac{1 + \epsilon}{1 - \epsilon}$ and note that for $\epsilon\le 1/2$, $\epsilon' \le 3\epsilon$. Given $x_{<t}\in[d]^{t-1}$, let $h(x_{<t})\in[d]^{d}$ denote the vector whose $j$-th entry is the number of occurrences of element $j\in[d]$ in $x_{<t}$. For any $h_1,h_2\in\N$, define \begin{equation}
	A^{h_1,h_2} \triangleq \frac{1}{2}\left((1-\epsilon)^{h_1}(1 + \epsilon)^{h_2} + (1 - \epsilon)^{h_2}(1 + \epsilon)^{h_1}\right)
\end{equation}
\begin{equation}
	B^{h_1,h_2} \triangleq \frac{1}{2}\left((1-\epsilon)^{h_1}(1 + \epsilon)^{h_2} - (1 - \epsilon)^{h_2}(1 + \epsilon)^{h_1}\right).
\end{equation}

\begin{fact}
	For any $x_{<t}\in[d]^{t-1}$,
	\begin{equation}
		\Delta(x_{<t}) = \prod^{d/2}_{a=1}A^{h(x_{<t})_{2a-1}, h(x_{<t})_{2a}}
	\end{equation}
	\begin{equation}
		\E[z,z'\sim\{\pm 1\}^{d/2}]*{\iprod{z,z'} \cdot \Psi^{z,z'}_{x_{<t}}} = \sum^{d/2}_{a=1}\left(B^{h(x_{<t})_{2a-1}, h(x_{<t})_{2a}}\right)^2\cdot \prod_{a'\neq a}\left(A^{h(x_{<t})_{2a-1}, h(x_{<t})_{2a}}\right)^2.
	\end{equation}
	\label{fact:ABs}
\end{fact}

We can now complete the proof of Theorem~\ref{thm:paninski}.

\begin{proof}[Proof of Theorem~\ref{thm:paninski}]
	We will show that as long as $t\le O(\sqrt{d}/\epsilon^2)$, $Z_t$ defined in \eqref{eq:paninski_main} is no greater than $O(\epsilon^2/\sqrt{d})$, from which the theorem follows. Fix any $t\le N$. Let $\Mul_t(U)$ denote the multinomial distribution over $d$-tuples $\vh$ for which $\sum^d_{i=1}h_i = t$. By Fact~\ref{fact:ABs} and \eqref{eq:innerexp}, \begin{equation}
		Z_t = \frac{2\epsilon^2}{d}\sum^{d/2}_{a=1}\E[\vh\sim \Mul_t(U)]*{\frac{(B^{h_{2a-1},h_{2a}})^2}{A^{h_{2a-1},h_{2a}}}\cdot \prod_{a'\neq a}A^{h_{2a'-1},h_{2a'}}} \triangleq \frac{2\epsilon^2}{d}\sum^{d/2}_{a=1}C^{(a)}_t.\label{eq:exp_term}
	\end{equation} Fix any $a\in[d]$; without loss of generality suppose $a = 1$. Then \begin{align}
		C^{(1)}_t &= \frac{1}{d^t}\sum_{t}\binom{\vh}{h_1\cdots h_d}\frac{(B^{h_1,\ell - h_1})^2}{A^{h_1,\ell - h_1}}\cdot \prod_{a'\neq 1}A^{h_{2a'-1},h_{2a'}} \\
		&= \frac{1}{d^t}\sum^t_{\ell =0}\sum^\ell_{h_1 = 0}\frac{t!}{h_1!(\ell - h_1)!(t - \ell)!}\frac{(B^{h_1,\ell - h_1})^2}{A^{h_1,\ell - h_1}} \sum_{h_3+\cdots+h_d = t - \ell}\binom{t - \ell}{h_3\cdots h_d}\prod_{a'\neq 1}A^{h_{2a'-1},h_{2a'}} \\
		&= \E[\ell\sim\Bin(t,2/d)]*{\E[h_1\sim\Bin(\ell,1/2)]*{\frac{(B^{h_1,\ell - h_1})^2}{A^{h_1,\ell - h_1}}}}.\label{eq:binom}
	\end{align}
	Next, note that for any $h_1,h_2$, \begin{equation}
		\frac{(B^{h_1,h_2})^2}{A^{h_1,h_2}} = A^{h_1,h_2} - \frac{2}{\left(\frac{1 + \epsilon}{1 - \epsilon}\right)^{h_1 - h_2} + \left(\frac{1 + \epsilon}{1 - \epsilon}\right)^{h_2 - h_1}}\le A^{h_1,h_2} - \exp\left(-(h_1 - h_2)^2\epsilon'^2/2\right).
	\end{equation} Clearly $\E[h_1\sim\Bin(\ell,1/2)]{A^{h_1,\ell - h_1}} = 1$, so for any $0\le \ell \le t$, \begin{align}
		\E[h_1\sim\Bin(\ell,1/2)]*{\frac{(B^{h_1,\ell - h_1})^2}{A^{h_1,\ell - h_1}}} &\le \E[h_1\sim\Bin(\ell,1/2)]*{1 - \exp\left(-(2h_1 - \ell)^2\epsilon'^2/2\right)} \\
		&\le \E[h_1\sim\Bin(\ell,1/2)]*{2(h_1 - \ell/2)^2\epsilon'^2} = \epsilon'^2 \cdot \ell,
	\end{align} where in the last step we used the expression for the variance of a binomial distribution. Substituting this into \eqref{eq:binom}, we conclude that $C^{(a)}_t \le 2\epsilon'^2 t/d \le 18\epsilon^2t/d$ for all $a\in[d]$, so for $t = O(\sqrt{d}/\epsilon^2)$, \eqref{eq:exp_term} is at most $O(\epsilon^2/\sqrt{d})$ as desired. 
\end{proof}


\section{Miscellaneous Technical Facts}
\label{app:facts}

\begin{fact}
	For any $x > 1, c > 0$, $\frac{2}{x^c + x^{-c}} \ge \exp\left(-\frac{c^2}{2}\log^2 x\right)$.
\end{fact}

\begin{fact}[Integration by parts]
	Let $a,b\in\R$. Let $Z$ be a nonnegative random variable satisfying $Z\le b$ and such that for all $x\ge a$, $\Pr{Z > x} \le \tau(x)$. Let $f: [0,b]\to\R_{\ge 0}$ be nondecreasing and differentiable. Then \begin{equation}
		\E{f(Z)} \le f(a)(1 + \tau(a)) + \int^b_a \tau(x) f'(x) \ \d\, x.
	\end{equation}
	\label{fact:stieltjes}
\end{fact}

\begin{proof}
	Let $g: [0,b]\to [0,1]$ denote the CDF of $Z$, so that for $x\ge a$, $1 - g(x) \le \tau(x)$. Then 
		\begin{align}
			\E{Z^n} &= \int^b_0 f(Z) \ \d\, g \le f(a) + \int^b_a f(Z) \ \d\, g \\
			&= f(a) + f(b) g(b) - f(a) g(a) - \int^b_a g(x) f'(x) \ \d\, x \\
			&= f(a) (1 - g(a)) + f(b) - (f(b) - f(a)) + \int^b_a (1 - g(x))f'(x)\ \d\, x \\
			&\le f(a)(1 + \tau(a)) + \int^b_a \tau(x) f'(x) \ \d \, x,
		\end{align} where the first integral is the Riemann-Stieltjes integral, the third step is integration by parts, the fourth step follows because $g(b) = 1$, and the last follows because $1 - g(x) \le \tau(x) \le 1$ for $x \ge a$.
\end{proof}

\end{document}